\documentclass[12pt]{article}

\usepackage[english]{babel}

\usepackage[letterpaper,top=2cm,bottom=2cm,left=3cm,right=3cm,marginparwidth=3cm]{geometry}

\usepackage{amsthm}
\usepackage{amsmath}
\usepackage{graphicx}
\usepackage[colorlinks=true, allcolors=blue]{hyperref}

\newtheorem{theorem}{Theorem}
\newtheorem{proposition}{Proposition}

\title{Computational Complexity of Envy-free and Exchange-stable Seat Arrangement Problems on Grid Graphs}
\author{
    Sota Kawase\thanks{School of Social Information Science, University of Hyogo, \texttt{fa21w022@guh.u-hyogo.ac.jp}} \and Shuichi Miyazaki\thanks{Graduate School of Information Science, University of Hyogo, \texttt{shuichi@sis.u-hyogo.ac.jp}}
}

\begin{document}
\maketitle
\begin{abstract}
The {\em Seat Arrangement Problem} is a problem of finding a desirable seat arrangement for given preferences of agents and a seat graph that represents a configuration of seats.
In this paper, we consider decision problems of determining if an envy-free arrangement exists and an exchange-stable arrangement exists, when a seat graph is an $\ell \times m$ grid graph.
When $\ell=1$, the seat graph is a path of length $m$ and both problems have been known to be NP-complete.
In this paper, we extend it and show that both problems are NP-complete for any integer $\ell \geq 2$.
\end{abstract}

\section{Introduction}
Recently, a variant of resource arrangement problems, called the {\em Seat Arrangement Problem}, is studied extensively \cite{2, 4}.
In this problem, we are given a set of agents $A$ and an undirected graph $G$, called a {\em seat graph}, where the number of agents in $A$ is equal to the number of vertices of $G$.
Each agent has preference to other agents, which is expressed by {\em utility}.
An agent $i$'s utility for another agent $j$ is a numerical value, which shows how much $i$ likes $j$.
The higher the value is, the more $i$ likes $j$.
Note that the utility may be negative, in which case it can be interpreted as ``$i$ dislikes $j$.''
The seat graph $G$ represents a configuration of seats, where vertices correspond to seats and two seats are neighbors to each other if and only if corresponding two vertices are adjacent in $G$.
The task of the problem is to map the agents to vertices of $G$ in one-to-one fashion.
This mapping is called an {\em arrangement}.
A {\em utility of an agent in an arrangement} is the sum of his utilities over all the agents placed on the neighboring seats and the {\em utility of an arrangement} is the sum of the utilities of all the agents.

There are several solution concepts \cite{3, 4}.
A {\em maximum welfare arrangement} is an arrangement with maximum utility.
A {\em maxmin utility arrangement} is one that maximizes the minimum utility among all the agents.
For an arrangement, we say that an agent $i$ {\em envies} an agent $j$ if $i$'s utility increases when $i$ and $j$ exchange their seats.
An {\em envy-free arrangement} is an arrangement in which no agent envies anyone.
If agents $i$ and $j$ envy each other, $(i, j)$ is called a {\em blocking pair}, and an {\em exchange-stable arrangement} is one that has no blocking pair.
%For these four solution concepts, the problems of deciding if there exists a desirable seat arrangement are denoted MWA, MUA, EFA, and ESA, respectively.
Note that an envy-free arrangement is always an exchange-stable arrangement, but the converse is not true.
%Corresponding to these solution concepts, we may consider decision problems that ask to determine if such an arrangement exists.
For these four solution concepts, the problems of deciding if there exists a desirable seat arrangement are denoted MWA, MUA, EFA, and ESA, respectively.
Among these four problems, we focus on EFA and ESA in this paper.
These problems are considered on several graph classes and shown to be NP-complete for very restricted classes of seat graphs, such as paths or cycles \cite{2, 4}.
These hardness are also shown for seat graphs being a composition of a graph $G$ with some number of isolated vertices, when $G$ belongs to some graph class such as stars, clique, or matching \cite{4}. 
For studying computational complexity of the problem in terms of seat graph classes, it is natural to attack from simpler ones. 
Then one of the next candidates might be grid graphs, which is also remarked in \cite{1, 2}.

{\bf Our contributions.} In this paper, we study computational complexity of EFA and ESA on $\ell \times m$ grid graphs.
For applications, a grid graph may be considered as a configuration of a classroom and in this context the problem corresponds to finding a seat arrangement of students in schools.
Note that when $\ell=1$, an $\ell \times m$ grid graph is just a simple path of length $m$, and both EFA and ESA are already shown to be NP-hard as mentioned above.
We extend these results and show that both problems are NP-complete for any integer $\ell \geq 2$.
%It should be noted that hardness for $\ell =1$ does not trivially imply hardness for $\ell \geq 2$. 
Our hardness proofs are inspired by those for path graphs \cite{2, 4}, which are reductions from the Hamiltonian path problem.
However, extensions are not straightforward because we needed several new ideas to care about existence of cycles in a grid graph, which was not necessary for path graphs.
Furthermore, for exchange-stability, we needed to craft utilities so that envy-freeness simulates exchange-stability.
For technical reasons it was necessary to prove the cases with $\ell=2$ and/or $\ell=3$ differently from general $\ell$.
Therefore, there may be several arguments that overlap in proofs, but this manuscript includes all of the full details, even if redundant, for the sake of completeness of each proof.

{\bf Related work.}
The pioneering study about the Seat Arrangement Problem is due to Bodlaender et al.~\cite{3}.
They study the price of stability and the price of fairness of a seat arrangement. 
They also study computational complexity and parameterized complexity of the above four problems with respect to the maximum size of connected components in a seat graph.
Berriaud et al.~\cite{2} study the problems for restricted classes of seat graphs, and as mentioned above, they show NP-hardness of EFA and ESA on paths and cycles.
On the positive side, they define a {\em class} as a set of agents who share a common preference function, and show that EFA and ESA on paths and cycles become polynomially solvable when the number of classes is a constant.
Ceylan et al.~\cite{4} give more refined classification on computational complexity, using parameterized complexity as well, from three aspects: seat graph classes (such as stars, cliques, and matchings with isolated vertices), problem-specific parameters (such as the number of isolated seats and maximum number of non-zero entries in preferences), and preference structure (such as non-negative preferences and symmetric preferences).
Aziz et al.~\cite{1} define a relaxed notion of the exchange-stability, called the {\em neighborhood stability}, in which no two agents assigned to adjacent vertices form a blocking pair.
They study the existence of a neighborhood stable arrangement and computational complexity of finding it when seat graphs are paths or cycles.
Rodriguez~\cite{6} introduces new notions of utility, called {\em B-utility} and {\em W-utility}, to define a utility of an agent for a seat arrangement.
B-utility (resp.~W-utility) of an agent $i$ is the utility of $i$ towards the best (resp.~worst) agent seated next to him.
Rodriguez provides algorithmic and computational complexity results of the above four problems with respective to these utilities (as well as the standard utility, which he calls {\em S-utility}) for some restricted class of preferences, including {\em 1-dimensional preference} that he introduces.
Wilczynski \cite{7} considers the Seat Assignment Problem under ordinal preferences. 
In this model, preference is specified not by a utility function but by a linear order of agents.
She investigates existence of an exchange-stable or a popular seat arrangement for simple seat graphs such as paths, cycles, and clusters.

\section{Preliminaries}

The following definitions are mostly taken from \cite{2, 4}. 
We consider the problem of assigning $n$ agents to $n$ seats.
These $n$ seats are represented as vertices of an undirected graph, called a {\it seat graph},  $G = (V, E)$ such that $|V|=n$.
Each agent $i \in A$ has a numerical {\it preference} toward other agents, which is expressed by a function $p_i : A\setminus \{ i \} \rightarrow R $, where $p_i(j)$ represents $i$'s {\it utility} for $j$. 
A {\it preference profile} $P$ is the set of all agents' preferences.
We say that a preference is {\it binary} when $(p_i(j))_{i,j\in A}\in \{0,1\}$.

An {\it arrangement} of $A$ on a seat graph $G=(V, E)$ is a bijection $\pi: A \rightarrow V$.
For a vertex $v \in V$, let $N_G(v) = \{u \in V \mid \{u, v\}\in E$\} be the set of neighbors of $v$.
For an arrangement $\pi$, we define agent $i$'s {\it utility} in $\pi$ as $U_i(\pi) = \sum_{\substack{v \in N_{G}(\pi (i)) } }p_i(\pi ^{-1}(v))$, namely, the sum of $i$'s utility for the agents who are sitting next to $i$ in $\pi$.
For an arrangement $\pi$ and two agents $i, j \in A$, let $\pi_{ij}$ be the arrangement obtained from $\pi$ by swapping $\pi(i)$ and $\pi(j)$, namely, $\pi_{ij}(i)=\pi(j)$, $\pi_{ij}(j)=\pi(i)$, and $\pi_{ij}(a)=\pi(a)$ for $a \in A \setminus \{ i, j \}$.
We say that $i$ {\it envies} $j$ in $\pi$ if $U_i(\pi_{ij}) > U_i(\pi)$ holds. 
If $i$ envies someone in $\pi$, we say that $i$ {\em has an envy} in $\pi$.
An arrangement $\pi$ is {\it envy-free} if no agent has an envy in $\pi$.
We say that $(i,j)$ is a {\it blocking pair} if $i$ and $j$ envies each other. 
An agent $i$ is called a {\em blocking agent} if there is a blocking pair $(i, j)$ for some agent $j$.
An arrangement $\pi$ is {\it exchange-stable} if there is no blocking pair in $\pi$.

For a directed graph $G$, a {\em directed Hamiltonian path} is a directed path in $G$ that passes through each vertex exactly once.
We may sometimes omit the word ``directed'' when it causes no confusion.
Let us denote by {\em DHP} the problem of determining if a given graph has a Hamiltonian path.
It is well known that DHP is NP-complete \cite{5}.
Let {\em DHP*} be a restriction of DHP where an input graph has a special vertex that has no outgoing edge and has incoming edges from all the other vertices.
We use the following fact in our hardness proofs.

\begin{proposition}
DHP* is NP-complete.
\end{proposition}

\begin{proof}
Membership in NP is obvious.
We prove the NP-hardness by a reduction from DHP.
Let $G=(V, E)$ be an instance of DHP.
We construct the directed graph $G' = (V', E')$ where $V' = V \cup \{ v^* \}$ and $E' = E \cup \{ (v, v^*) \mid v \in V\}$ as an input to DHP*. 
It is easy to see that $G'$ satisfies the condition of DHP*.

Suppose that there is a Hamiltonian path $H$ in $G$.
Then, a path obtained by appending $v^*$ to the tail of $H$ is a Hamiltonian path in $G'$.
Conversely, if there is a Hamiltonian path $H'$ in $G'$, its last vertex must be $v^*$ because it has no outgoing edge.
Therefore, removing $v^*$ from $H'$ would give us a Hamiltonian path in $G$.
\end{proof}

In the following sections, when we use DHP{*} as a reduction source, vertices are denoted $v_1, v_2, \ldots, v_n$, and without loss of generality we assume that $v_n$ has no outgoing edge.

\section{Envy-freeness}

\begin{theorem}
Deciding whether an envy-free arrangement on a $2\times m$ grid graph exists is NP-complete even for binary preferences. 
\end{theorem}

\begin{proof}
Membership in NP is obvious.
We prove the hardness by a reduction from DHP*. Let $G=(V, E)$ be an instance of HP and let $n=|V|$. For each vertex $v_i \in V (1 \leq i \leq n)$, we introduce six agents $x_i$, $y_i$, $z_i$, $a_i$, $b_i$, and $c_i$, hence there are $6n$ agents in total. Their preference profile $P$ is given in Table \ref{table:1}. 
In the table, each row represents preference of an agent given at the leftmost column.
For example, $p_{x_i}(y_i) = 1$ and $p_{x_i}(s) = 0$ for $s\in A\setminus \{y_i\}$.
Note that $p_{z_n}(a)=0$ for any $a$ because $v_n$ has no outgoing edge. 
Finally, we set $m = 3n$, i.e., our seat graph is a $2 \times (3n)$ grid graph.

\begin{table}[t]
  \centering
  
    \begin{tabular}{c|c|c}\hline\hline
       & 1 & 0 \\ \hline
      $x_i(1 \leq i \leq n)$ & $y_i$ & other agents \\ \hline
      $y_i(1 \leq i \leq n)$ & $z_i$ & other agents \\ \hline
      $z_i(1 \leq i \leq n)$ & $x_p ((v_i, v_p)\in E)$ & other agents \\ \hline
    \end{tabular}
    
    \vspace{3mm}

  \centering
    \begin{tabular}{c|c|c}\hline\hline
       & 1 & 0 \\ \hline
      $a_i(1 \leq i \leq n)$ & $x_i$ & other agents \\ \hline
      $b_i(1 \leq i \leq n)$ & $y_i$ & other agents \\ \hline
      $c_i(1 \leq i \leq n)$ & $z_i$ & other agents \\ \hline
     
    \end{tabular}
    \caption{Preference profile $P$}\label{table:1}
    
\end{table}

We first argue that if there is a Hamiltonian path in $G$, $P$ has an envy-free arrangement on a $2 \times m$ grid graph.
Suppose that $G$ has a Hamiltonian path $H$.
Recall that $v_n$ has no outgoing edge, so $v_n$ must be the last vertex of $H$.
By renaming other vertices, we assume without loss of generality that $H$ is ordered according to the indices of vertices, i.e., $H:= v_1, v_2, \ldots, v_n$.
Corresponding to $H$, define the seat arrangement $\pi$ as shown in Figure \ref{fig:1}. 

\begin{figure}[htbp]
\begin{center}
\includegraphics[width=120mm]{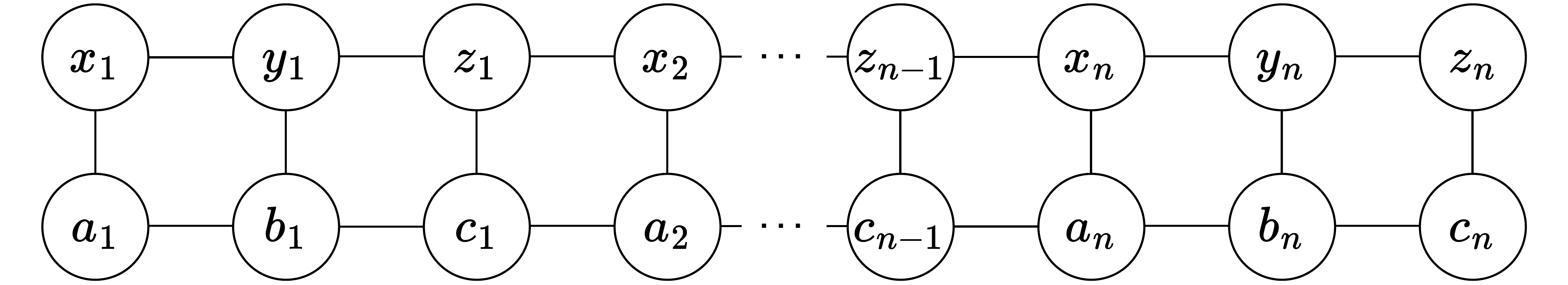}
\caption{The seat arrangement $\pi$}\label{fig:1}
\end{center}
\end{figure}

We will show that $\pi$ is an envy-free arrangement. First, observe from the definition of $P$ that the maximum possible utility for each of $x_i$, $y_i$, $a_i$, $b_i$, and $c_i$ is 1, and all these agents get utility 1 in $\pi$. Hence these $5n$ agents have no envy. 
Next, $z_n$ has no envy because his utility is 0 no matter where he sits. 
Finally, $z_i (i \neq n)$ gets utility 1 in $\pi$, but he can never get more utility because there is no vertex (seat) that is adjacent to two $x_j$s.
Thus $z_i$ has no envy.
Therefore, nobody has envy and hence $\pi$ is envy-free.

For the opposite direction, assume that there is an envy-free arrangement $\pi$ of $P$ on a $2 \times m$ grid graph. We will show that $G$ has a Hamiltonian path. First, we claim that for any $i$, $a_i$ is placed next to $x_i$ in $\pi$, as otherwise $a_i$ gets utility 0 now and so envies an agent who is located next to $x_i$.
For the same reason, $b_i$ is placed next to $y_i$, $c_i$ is placed next to $z_i$, $x_i$ is placed next to $y_i$, and $y_i$ is placed next to $z_i$.
Next, let us consider $z_{i} (i \neq n)$. Since $v_{i}$ has at least one outgoing edge in $G$, $z_i$ is placed next to some $x_p$ such that $(v_i, v_p)\in E$, as otherwise, $z_{i}$ envies $x_{p}$'s neighbor.
However, $z_{i}$ can never have two $x_{j}$s as his neighbor because the maximum degree of the seat graph is three and $y_i$ and $c_i$ must be neighbors of $z_i$ as mentioned above. Therefore, $z_i (i \neq n)$ has exactly one neighbor $x_p$ such that $(v_i, v_p)\in E$.
Similarly, since $x_i$ must have $y_i$ and $a_i$ as his neighbors, $x_i$ can have at most one $z_j$ as a neighbor.

Now, from $\pi$, construct the directed graph $G'=(V', E')$ where $V'= \{ u_1, u_2, \ldots, u_n \}$ and there is an arc $(u_{i}, u_{j}) \in E'$ if and only if $z_i$ is placed next to $x_j$ in $\pi$. From the above observations, we know that each $u_{i} (i \neq n)$ has exactly one outgoing edge and each $u_{i} (1 \leq i \leq n)$ has at most one incoming edge. These facts imply that $G'$ consists of at most one directed path and some number of directed cycles. In the following, we will show that there is no cycle. 
Assume on the contrary that there exists a cycle $u_{s_{1}}, u_{s_{2}}, \ldots, u_{s_{t}}, u_{s_{1}}$ of length $t$. Then, in the seat arrangement $\pi$, $z_{s_{i}}$ is placed next to $x_{s_{i+1}}$ for $1 \leq i \leq t-1$ and $z_{s_{t}}$ is placed next to $x_{s_{1}}$, so there exists a cycle $x_{s_{1}}, y_{s_{1}}, z_{s_{1}}, x_{s_{2}}, y_{s_{2}}, z_{s_{2}}, \dots, x_{s_{t}}, y_{s_{t}}, z_{s_{t}}, x_{s_{1}}$ in the seat graph.
Since the seat graph is a $2 \times m$ grid, there must be two consecutive bends as depicted in Cases (i), (ii) and (iii) of Figure \ref{fig:2}. Case (i) is impossible because $c_{s_{k}}$ is not placed next to $z_{s_{k}}$ as opposed to the above observation. For the same reason, neither Case (ii) nor (iii) is possible. Hence we can exclude the existence of a cycle and conclude that $G'$ consists of one path $u_{s_{1}}, u_{s_{2}}, \ldots, u_{s_{n}}$.
By construction of $G'$, for each $i (1\leq i \leq n-1$), $z_{s_{i}}$ is placed next to $x_{s_{i+1}}$ in $\pi$. Thus, by the aforementioned property, there is an arc $(v_{s_{i}}, v_{s_{i+1}})$ in $G$ and hence there is a Hamiltonian path $v_{s_{1}}, v_{s_{2}}, \ldots, v_{s_{n}}$. This completes the proof.

\begin{figure}[htbp]
\begin{center}
\includegraphics[width=120mm]{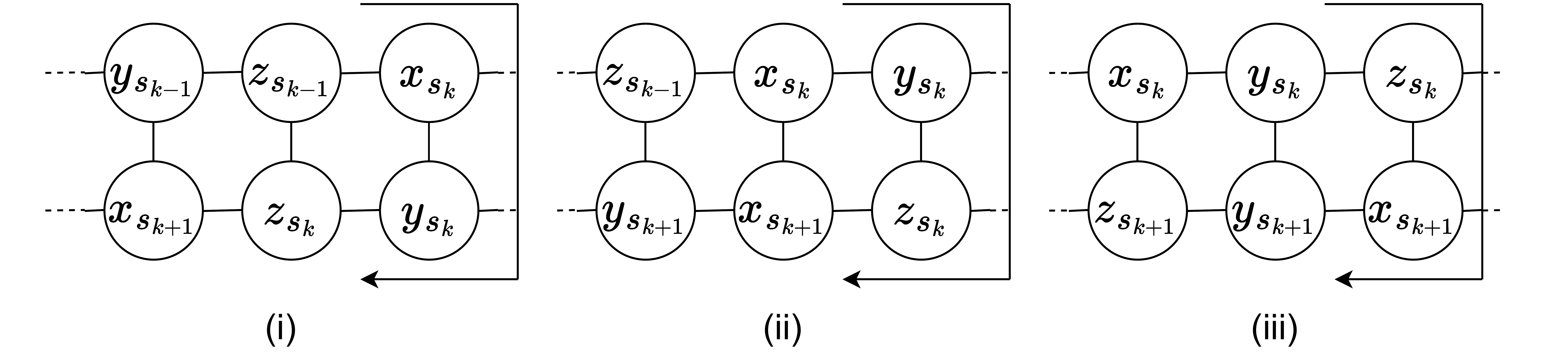}
\caption{Three possible configurations of the seat arrangement $\pi$ when there is a cycle in $G'$}\label{fig:2}
\end{center}
\end{figure}

\end{proof}

\begin{theorem}
Deciding whether an envy-free arrangement on an $\ell \times m$ grid graph exists is NP-complete for $\ell \geq 3$ even for binary preferences.
\end{theorem}
    
\begin{proof}
Membership in NP is obvious.
We prove the hardness by a reduction from DHP*. Let $G = (V, E)$ be an instance of DHP* and let $n = |V|$. For each vertex $v_i \in V (1 \leq i \leq n)$, we introduce nine agents $x_i$, $y_i$, $z_i$, $a_i$, $b_i$, $c_i$, $d_i$, $e_i$, and $f_i$. 
We further add {\it dummy} agents $D_i$ $(1 \leq i \leq 3n\ell-9n)$. 
Therefore, there are $3n\ell$ agents in total. 
Note that when $\ell=3$, there are no dummy agents.
Their preference profile $P$ is given in Table \ref{table:2}. Finally, we set $m = 3n$, i.e., our graph is an ${\ell} \times (3n)$ grid graph.

\begin{table}[t]
    \centering
    
    \begin{tabular}{c|c|c}\hline\hline
       & 1 & 0 \\ \hline
      $x_i(1 \leq i \leq n)$ & $y_i$ & other agents \\ \hline
      $y_i(1 \leq i \leq n)$ & $z_i$ & other agents \\ \hline
      $z_i(1 \leq i \leq n)$ & $x_p ((v_i, v_p)\in E)$ & other agents \\ \hline
    \end{tabular}

  \vspace{3mm}  
  
    \centering
    \begin{tabular}{c|c|c}\hline\hline
       & 1 & 0 \\ \hline
      $a_i(1 \leq i \leq n)$ & $x_i$ & other agents \\ \hline
      $b_i(1 \leq i \leq n)$ & $y_i$ & other agents \\ \hline
      $c_i(1 \leq i \leq n)$ & $z_i$ & other agents \\ \hline
      $d_i(1 \leq i \leq n)$ & $x_i$ & other agents \\ \hline
      $e_i(1 \leq i \leq n)$ & $y_i$ & other agents \\ \hline
      $f_i(1 \leq i \leq n)$ & $z_i$ & other agents \\ \hline
    \end{tabular}
    \vspace{3mm}

    \centering
    \begin{tabular}{c|c}\hline\hline
        & 0 \\ \hline
      $D_i(1 \leq i \leq 3n\ell-9n)$ & all agents \\ \hline
    \end{tabular}
    \caption{Preference profile $P$}\label{table:2}

\end{table}

We first argue that if there is a Hamiltonian path in $G$, $P$ has an envy-free arrangement on an $\ell \times m$ grid graph. Suppose that $G$ has a Hamiltonian path $H$. Recall that $v_n$ has no outgoing edge, so $v_n$ must be the last vertex of $H$. By renaming other vertices, we assume without loss of generality that $H$ is ordered according to the indices of vertices, i.e., $H:= v_1, v_2, \ldots, v_n$.
Corresponding to $H$, define the seat arrangement $\pi$ as shown in Figure \ref{Fig:3}.

\begin{figure}[htbp]
\begin{center}
\includegraphics[width=100mm]{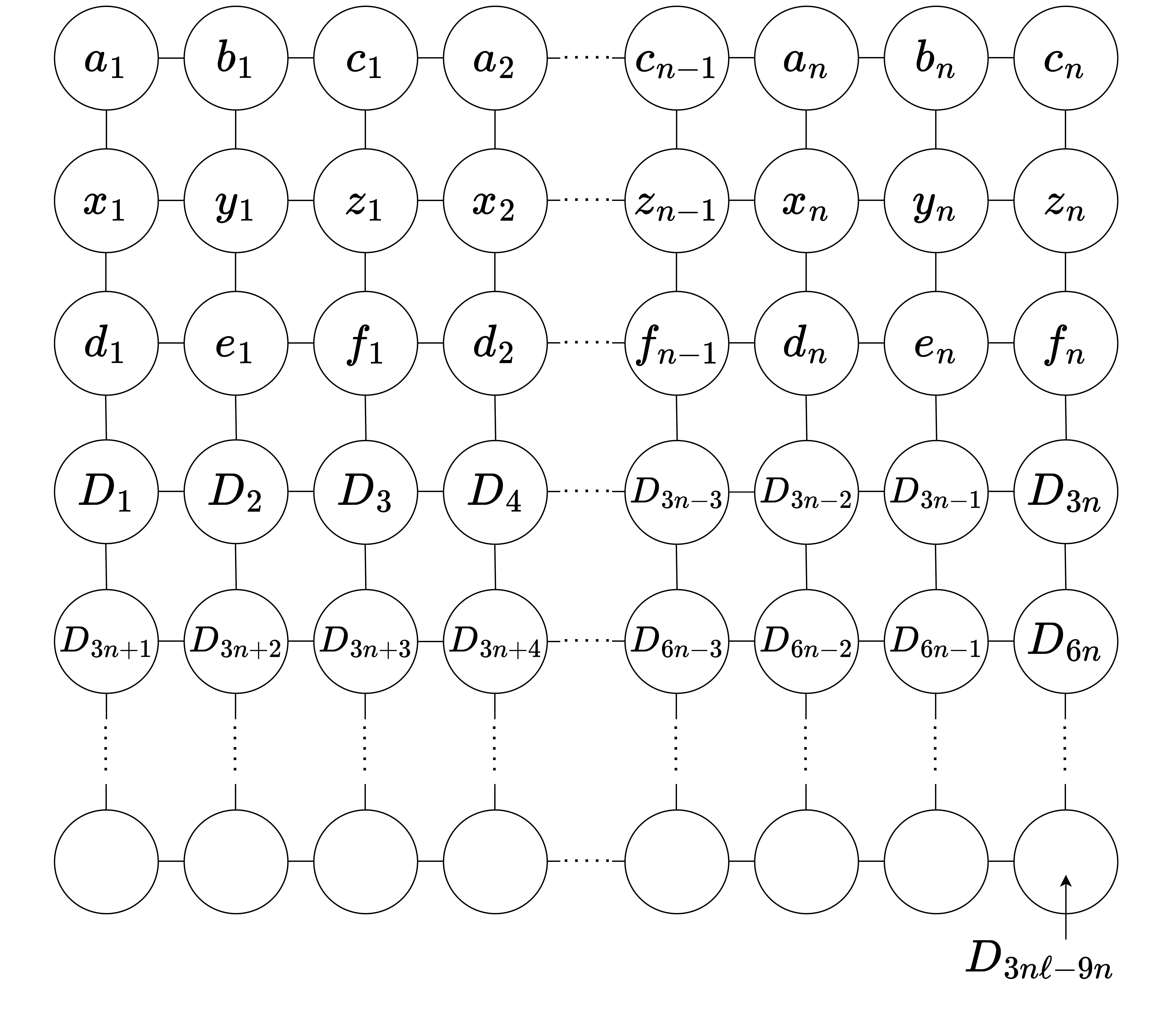}
\caption{The seat arrangement $\pi$}\label{Fig:3}
\end{center}
\end{figure}

We will show that $\pi$ is an envy-free arrangement.
First, observe from the definition of $P$ that the maximum possible utility for each of $x_i$, $y_i$, $a_i$, $b_i$, $c_i$, $d_i$, $e_i$, and $f_i$ is $1$, and all these agents get utility $1$ in $\pi$. 
Hence these $8n$ agents have no envy. 
Next, $z_n$ has no envy because his utility is 0 no matter where he sits.
Furthermore, $z_i (i \neq n)$ gets utility $1$ in $\pi$, but he can never get more utility because there is no vertex that is adjacent to more than one $x_j$.
Thus $z_i$ has no envy.
Finally, $D_i (1 \leq i \leq 3n\ell-9n)$ has no envy for the same reason for $z_n$.
Therefore, nobody has envy and hence $\pi$ is an envy-free arrangement.

For the opposite direction, assume that there is an envy-free arrangement $\pi$ of $P$ on an $\ell \times m$ grid graph. We will show that $G$ has a Hamiltonian path. First, we claim that for any $i$, $a_i$ and $d_i$ are placed next to $x_i$ in $\pi$, as otherwise they get utility $0$ now and so envy an agent who is located next to $x_i$ in $\pi$. 
For the same reason, $b_i$ and $e_i$ are placed next to $y_i$, $c_i$ and $f_i$ are placed next to $z_i$, $x_i$ is placed next to $y_i$, and $y_i$ is placed next to $z_i$.
Next, let us consider $z_i (i \neq n)$. Since $v_i$ has at least one outgoing edge in $G$, $z_i$ is placed next to some $x_p$ such that $(v_i, v_p) \in E$, as otherwise, $z_i$ envies $x_p$'s neighbor. However, $z_i$ can never have more than one $x_p$ as his neighbor because the maximum degree of the seat graph is four and $y_i$, $c_i$, and $f_i$ must be placed next to $z_i$. Therefore, $z_i (i \neq n)$ has exactly one neighbor $x_p$ such that $(v_i, v_p) \in E$. Similarly, since $x_i$ must have $y_i$, $a_i$, and $d_i$ as his neighbors, $x_i$ can have at most one $z_j$ as his neighbor.

Now, from $\pi$, construct the directed graph $G'=(V', E')$ where $V'= \{ u_1, u_2, \ldots, u_n \}$ and there is an arc $(u_{i}, u_{j}) \in E'$ if and only if $z_i$ is placed next to $x_j$ in $\pi$. From the above observations, we know that each $u_{i} (i \neq n)$ has exactly one outgoing edge and each $u_{i} (1 \leq i \leq n)$ has at most one incoming edge. These facts imply that $G'$ consists of at most one directed path and some number of directed cycles. In the following, we will show that there is no cycle. On the contrary, assume that there exists a cycle $u_{s_{1}}, u_{s_{2}}, \ldots, u_{s_{t}}, u_{s_{1}}$ of length $t$. 
Then, in the seat arrangement $\pi$, $z_{s_{i}}$ is placed next to $x_{s_{i+1}}$ for $1 \leq i \leq t-1$ and $z_{s_{t}}$ is placed next to $x_{s_{1}}$, so there exists a cycle $x_{s_{1}}, y_{s_{1}}, z_{s_{1}}, x_{s_{2}}, y_{s_{2}}, z_{s_{2}}, \dots, x_{s_{t}}, y_{s_{t}}, z_{s_{t}}, x_{s_{1}}$ in the seat graph.
This implies that there exists a vertex along this cycle where the cycle changes direction, as depicted in Cases (i), (ii) and (iii) of Figure \ref{fig:4}. 
Let us consider Case (i). Recall that $y_{s_{k-1}}$, $c_{s_{k-1}}$ and $f_{s_{k-1}}$ must be neighbors of $z_{s_{k-1}}$, so they are placed at vertices 1, 2, and 3.
Similarly, since $z_{s_{k}}$, $b_{s_{k}}$, and $e_{s_{k}}$ must be neighbors of $y_{s_{k}}$, they are placed at vertices 3, 4, and 5.
However, it is impossible that these six people are placed at five vertices.
Hence we can exclude Case (i).
Cases (ii) and (iii) can also be excluded in the same manner. 
Hence we can exclude the existence of a cycle and conclude that $G'$ consists of one path $u_{s_{1}}, u_{s_{2}}, \ldots, u_{s_{n}}$.
By construction of $G'$, for each $i (1\leq i \leq n-1$), $z_{s_{i}}$ is placed next to $x_{s_{i+1}}$ in $\pi$. Thus, by the aforementioned property, there is an arc $(v_{s_{i}}, v_{s_{i+1}})$ in $G$ and hence there is a Hamiltonian path $v_{s_{1}}, v_{s_{2}}, \ldots, v_{s_{n}}$. This completes the proof.
\begin{figure}[htbp]
\begin{center}
\includegraphics[width=150mm]{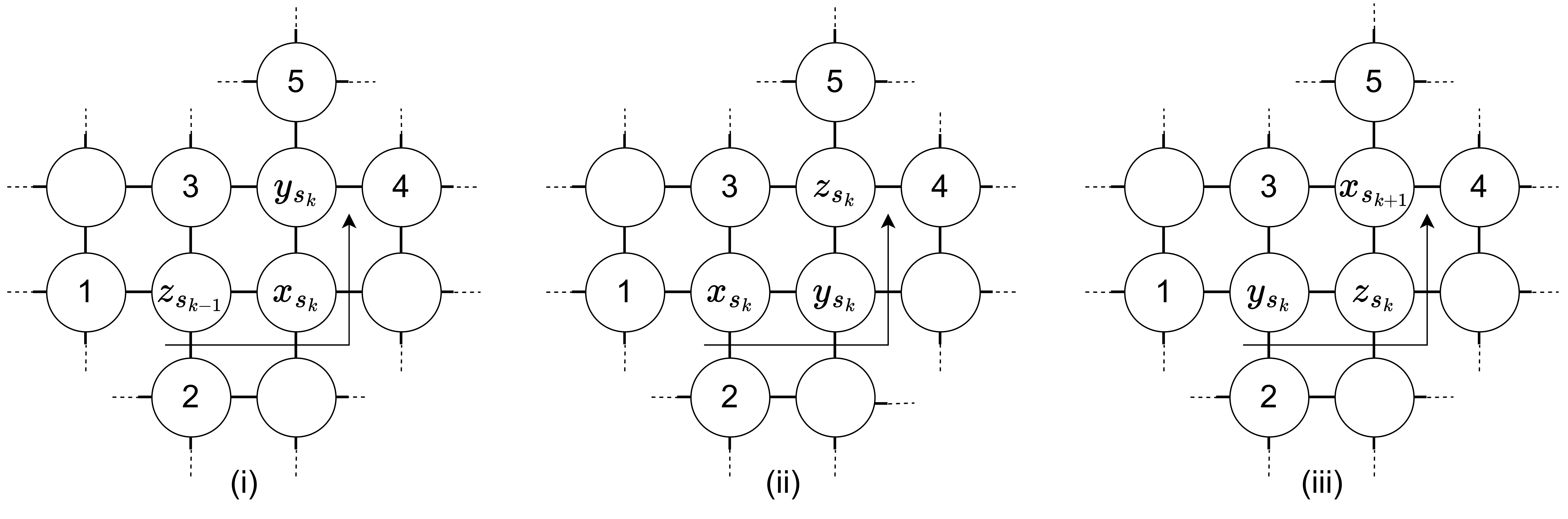}
\caption{Three possible configurations of the seat arrangement $\pi$ when there is a cycle in $G'$}\label{fig:4}
\end{center}
\end{figure}

\end{proof}

\section{Exchange-stability}

\begin{theorem}
Deciding whether an exchange-stable arrangement on a $2\times m$ grid graph exists is NP-complete. 
\end{theorem}

\begin{proof}
Membership in NP is obvious.
We prove the hardness by a reduction from DHP*. Let $G = (V, E)$ be an instance of DHP* and let $n = |V|$. For each vertex $v_i \in V$, introduce six agents $x_i$, $y_i$, $z_i$, $a_i$, $b_i$, and $c_i$. Moreover, we add four agents $s$, $t_1$, $t_2$, and $t_3$. Therefore, there are $6n+4$ agents in total. Their preference profile $P$ is given in Table \ref{table:3}. Finally, we set $m = 3n+2$, i.e., our seat graph is a $2 \times (3n+2)$ grid graph.

\begin{table}[t]
    
    \centering
    
    \begin{tabular}{c|c|c|c|c}\hline\hline
       & $-10$ & $-1$ & $0$ & $3$ \\ \hline
      $x_i(1 \leq i \leq n)$ & $s$ & other agents & $t_1$, $t_2$, $t_3$ & $y_i$ \\ \hline
      $y_i(1 \leq i \leq n)$ & $s$ & other agents & $t_1$, $t_2$, $t_3$ & $z_i$ \\ \hline
      $z_i(1 \leq i \leq n-1)$ & $s$ & other agents & $t_1$, $t_2$, $t_3$ & $x_p$ $((v_i, v_p)\in E)$ \\ \hline
      $z_n$ & $s$ & other agents & $t_1$, $t_2$, $t_3$ & $c_n$  \\ \hline
    \end{tabular} 
    
\vspace{3mm}   

    \begin{tabular}{c|c|c|c|c}\hline\hline
       & $-10$ & $-1$ & $0$ & $3$ \\ \hline
      $a_i(1 \leq i \leq n)$ & $s$ & other agents & $t_1$, $t_2$, $t_3$ & $x_i$ \\ \hline
      $b_i(1 \leq i \leq n)$ & $s$ & other agents & $t_1$, $t_2$, $t_3$ & $y_i$ \\ \hline
      $c_i(1 \leq i \leq n)$ & $s$ & other agents & $t_1$, $t_2$, $t_3$ & $z_i$ \\ \hline
    \end{tabular}

\vspace{3mm}

    \centering
    \begin{tabular}{c|c|c}\hline\hline
       & $0$ & $1$ \\ \hline
      $s$ & $t_1$, $t_2$, $t_3$ & other agents   \\ \hline
      $t_i(1 \leq i \leq 3)$ & other agents & $s$  \\ \hline
    \end{tabular}
    \caption{Preference profile $P$}\label{table:3}
  
\end{table}

We first argue that if there is a Hamiltonian path in $G$, $P$ has an exchange-stable arrangement on a $2 \times m$ grid graph. Suppose that $G$ has a Hamiltonian path $H$. Recall that $v_n$ has no outgoing edge, so $v_n$ must be the last vertex of $H$. By renaming other vertices, we assume without loss of generality that $H$ is ordered according to the indices of vertices, i.e., $H:= v_1, v_2, \ldots, v_n$.
Corresponding to $H$, define the seat arrangement $\pi$ as shown in Figure \ref{fig:5}, where each agent's utility in $\pi$ is given above or below the vertex he is seated.

\begin{figure}[htbp]
\begin{center}
\includegraphics[width=140mm]{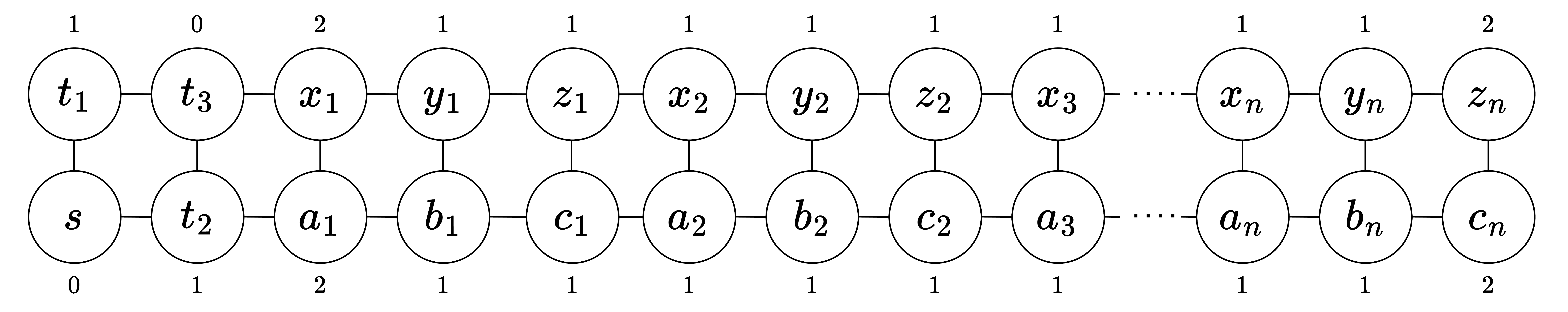}
\caption{The seat arrangement $\pi$}\label{fig:5}
\end{center}
\end{figure}

We will show that $\pi$ is an exchange-stable arrangement.
To facilitate the proof, we partition the agents into three groups, and for each group we show that no one can be a blocking agent.
For an agent $q \in \{ x_i, y_i, z_i, a_i, b_i, c_i \mid 1 \leq i \leq n \}$, $q$'s {\em favorite} is an agent to whom $q$'s utility is 3. 
For example, $a_3$'s favorite is $x_3$ and $y_5$'s favorite is $z_5$. Note that $z_i (1 \leq i \leq n-1)$ may have more than one favorite.

{\bf\boldmath Group 1. $s$, $t_1$, $t_2$, and $t_3$.}
First, note that both $t_1$ and $t_2$ already have maximum possible utility in $\pi$.
Hence they have no envy in $\pi$ and so cannot be a blocking agent. 
Agent $t_3$ has utility 0 now, and he can get positive utility only by sitting next to $s$. Hence he envies only $t_1$ and $t_2$.
However, neither $t_1$ nor $t_2$ has envy in $\pi$ as mentioned above, so $t_3$ cannot be a blocking agent.
Agent $s$'s neighbors are $t_1$ and $t_2$, but utilities of agents $x_i$, $y_i$, $z_i$, $a_i$, $b_i$, and $c_i$ ($1 \leq i \leq n$) for $t_1$ and $t_2$ are 0, so their utilities decrease to 0 by moving to $\pi(s)$, that is, they do not envy $s$. Hence $s$ cannot be a blocking agent.
Thus no one in Group 1 is a blocking agent.

{\bf\boldmath Group 2. $x_i$, $y_i$, $a_i$, $b_i$, and $c_i$ $(1 \leq i \leq n)$.}
Note that these agents have only one favorite and are now sitting next to the favorite.
Therefore, to increase the utility, an agent still has to sit next to his favorite after the swap.
For example, $y_2$'s favorite is $z_2$, so $y_2$ may envy only $c_2$, $x_3$, and $z_2$. But in this case, $y_2$ actually has no envy because his utility remains 1 by swapping a seat with any one of these three.
Checking in this way, we can see that Group 2 agents having indices between 2 and $n-1$ have no envy in $\pi$, and envies emanating from Group 2 agents are only the following: 
(i) $a_1$ envies $t_3$, (ii) $b_1$ envies $x_1$, (iii) $x_n$ envies $z_n$, (iv) $y_n$ envies $c_n$, and (v) $b_n$ envies $z_n$.
For (i), we already know that $t_3$ is not a blocking agent.
For (ii) and (iv), $x_1$ and $c_n$ are Group 2 agents and they are not listed in (i) through (v), so  they have no envy.
To deal with (iii) and (v), consider $z_n$.
$z_n$'s favorite is $c_n$ so he may envy only $b_n$ or $c_n$ (and so has no envy to $x_n$). 
Also, $z_n$'s utility decreases if he swaps a seat with $b_n$, so $z_n$ does not envy $b_n$.
Therefore, we can conclude that no one in Group 2 can be a blocking agent.

{\bf\boldmath Group 3. $z_i$ $(1 \leq i \leq n)$.}
We now have only to care about a blocking pair within Group 3 agents.
In the argument for Group 2, we already examined that $z_n$ does not envy any $z_j$.
For $z_i$ and $z_j$ ($1 \leq i, j \leq n-1$), they both have utility 1 now, but it does not increase by swapping the seats.
Thus no one in Group 3 is a blocking agent.

From the above discussion, there is no blocking pair in $\pi$, so $\pi$ is exchange-stable.

For the opposite direction, assume that there is an exchange-stable arrangement $\pi$ of $P$ on a $2\times m$ grid graph. We will show that $G$ has a Hamiltonian path. 
We say that an agent $p$ is {\em isolated} if $p$'s neighbors in $\pi$ are included in $\{t_1, t_2, t_3\}$.
First we show that $s$ is isolated.
Suppose not and let $p \not \in \{t_1, t_2, t_3\}$ be an agent who is seated next to $s$.
Since the maximum degree of the seat graph is three, there must be an agent $t \in \{t_1, t_2, t_3\}$ who is {\em not} seated next to $s$.
Note that $p$'s utility is now at most $-4$, but if he moves to $t$'s seat, the utility becomes at least $-3$. 
Similarly, $t$'s utility is $0$ now, but if he moves to $p$'s seat, it increases to $1$. 
These mean that $(p, t)$ is a blocking pair, contradicting the stability of $\pi$.
Therefore, $s$ is isolated.
Note that there can be at most one isolated agent in $A \setminus \{ s \}$, as shown in Figure \ref{fig:6}, where $a_1$ is isolated.

\begin{figure}[htbp]
\begin{center}
\includegraphics[width=100mm]{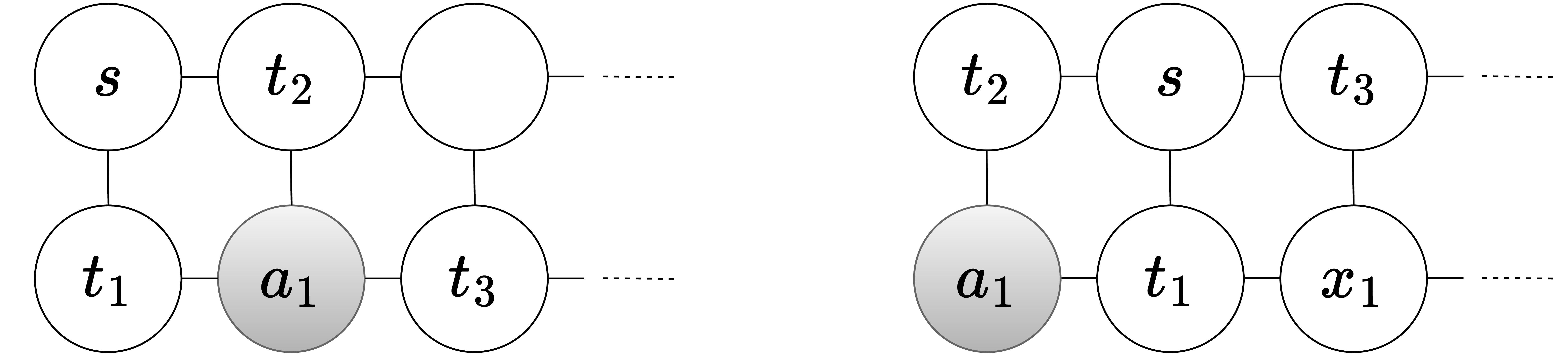}
\caption{Examples of $\pi$ including one isolated agent in $A \setminus \{ s \}$}\label{fig:6}
\end{center}
\end{figure}

Before proceeding with the proof, it would be helpful to explain the role of agent $s$.
Note that $s$'s utility in $\pi$ is 0 since he is isolated, but his utility becomes positive if he is not isolated.
Thus, he envies almost everyone.
In the following, we will show several conditions that an exchange-stable arrangement $\pi$ must satisfy, such as ``$x_i$ is placed next to $y_i$.''
To do so, we show that $x_i$ envies $s$ if this condition is not satisfied, implying that $(s, x_i)$ is a blocking pair.

Now, we continue the proof.
We claim that $x_i (1 \leq i \leq n)$ is not isolated for the following reason: 
If $x_i$ is isolated, then $a_i$ is not a neighbor of $x_i$.
Since $a_i$ is not isolated (as there can be at most one isolated agent other than $s$), $a_i$ has at least one neighbor to whom $a_i$'s utility is $-1$.
Thus $a_i$ can increase utility to 0 by moving to $s$'s seat.
Since $s$ is isolated, his utility is 0, but he can have a positive utility by moving to $a_i$'s seat since $a_i$ is not isolated.
Therefore $(a_i, s)$ is a blocking pair, contradicting the stability of $\pi$.
For the same reason, we can see that none of $y_i (1 \leq i \leq n)$, $z_i (1 \leq i \leq n)$, and $c_n$ are isolated.

Moreover, we claim that if $a_i (1 \leq i \leq n)$ is not isolated, $a_i$ is placed next to $x_i(1 \leq i \leq n)$, for if not, by swapping $a_i$ and $s$, $a_i$ can increase the utility from at most $-1$ to $0$ and $s$ can increase the utility from $0$ to at least $1$. Therefore, $(a_i, s)$ is a blocking pair, a contradiction. 
For the same reason, for $1 \leq i \leq n$, $b_i$ (resp.~$c_i$) is placed next to $y_i$ (resp.~$z_i$) if $b_i$ (resp.~$c_i$) is not isolated. 
Also, for the same reason, we can see that $x_i$ is placed next to $y_i$ and $y_i$ is placed next to $z_i$.

Then, let us consider $z_i (1 \leq i \leq n-1)$. Since the vertex $v_i$ has at least one outgoing edge in $G$, $z_i$ is placed next to some $x_p$ such that $(v_i, v_p) \in E$.
Suppose not. Then, by swapping $z_i$ and $s$, $z_i$ can increase the utility from at most $-1$ to $0$ and $s$ can increase the utility from $0$ to at least $1$. Therefore, $(z_i, s)$ is a blocking pair, a contradiction. 
We also show that $z_{i}$ can never have two $x_{p}$s as his neighbor. 
Assume on the contrary that $x_{p}$ and $x_q$ are $z_i$'s neighbors. 
Recall that $y_i$ must also be a neighbor of $z_i$, so $z_i$ is placed at a vertex of degree 3, as in Figure \ref{fig:7}.  
Recall that $x_i$ (reps. $y_p$ and $y_q$) must be neighbors of $y_i$ (reps. $x_p$ and $x_q$).
Also, if $b_i$ (reps. $a_p$ and $a_q$) is not isolated, he must be neighbors of $y_i$ (reps. $x_p$ and $x_q$).
As mentioned before, at most one of $b_i$, $a_p$, and $a_q$ may be isolated, but even so, the other five agents must be placed at four vertices 1 through 4. 
This is impossible, so we can exclude this configuration.
There are several more cases according whether how $y_i$, $x_p$, and $x_q$ are placed around $z_i$, but it is easy to see that none of them is possible by the same reasoning.

\begin{figure}[htbp]
\begin{center}
\includegraphics[width=70mm]{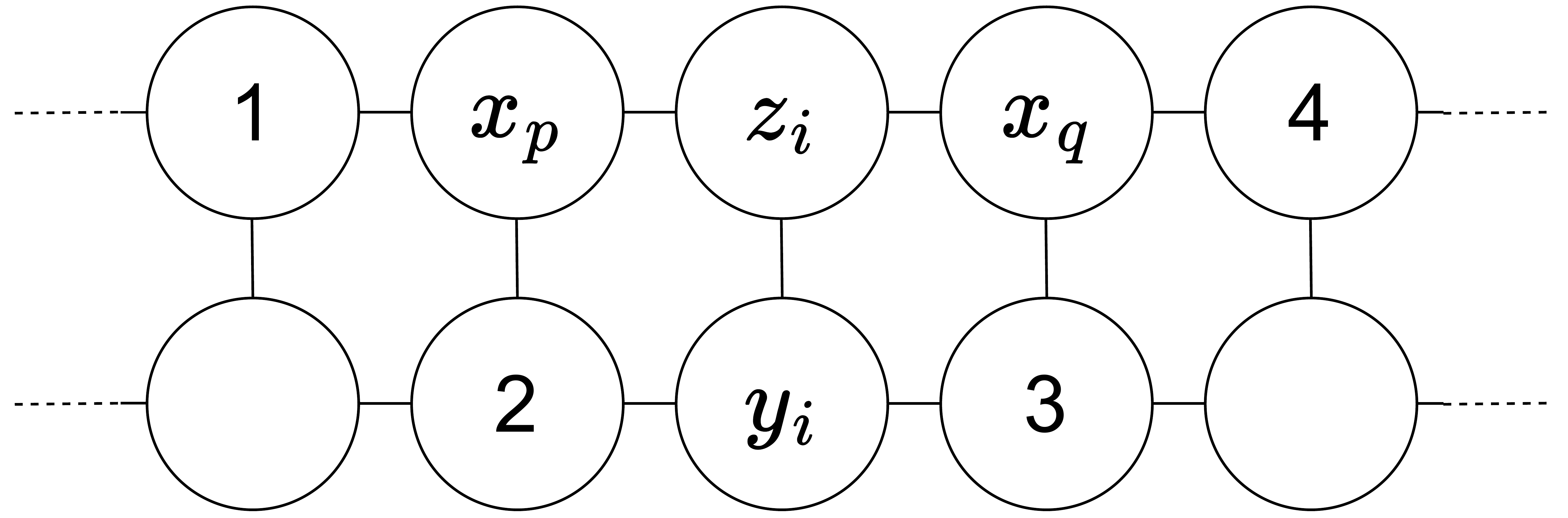}
\caption{One possible configuration of $y_j$, $x_p$, and $x_q$}\label{fig:7}
\end{center}
\end{figure}

Next, we show that $x_i (1 \leq i \leq n)$ can have at most one $z_p$ as a neighbor.
This can be shown in the same manner as for $z_i$ above, by assuming that $z_p$ and $z_q$ are neighbors of $x_i$ and by noting that
\begin{itemize}
\item $y_i$ must be a neighbor of $x_i$,

\item $b_i$ and $z_i$ must be neighbors of $y_i$,

\item $y_p$ and $c_p$ must be neighbors of $z_p$, and

\item $y_q$ and $c_q$ must be neighbors of $z_q$.

\end{itemize}
 
Now, from $\pi$, construct the directed graph $G'=(V', E')$ where $V'= \{ u_1, u_2, \ldots, u_n \}$ and there is an arc $(u_{i}, u_{j}) \in E'$ if and only if $z_i$ is placed next to $x_j$ in $\pi$. From the above observations, we know that each $u_{i} (1 \leq i \leq n-1)$ has exactly one outgoing edge and each $u_{i} (1 \leq i \leq n)$ has at most one incoming edge. These facts imply that $G'$ consists of at most one directed path and some number of directed cycles. In the following, we will show that there is no cycle. 
Assume on the contrary that there exists a cycle $u_{s_{1}}, u_{s_{2}}, \ldots, u_{s_{t}}, u_{s_{1}}$ of length $t$. 
Then, in the seat arrangement $\pi$, $z_{s_{i}}$ is placed next to $x_{s_{i+1}}$ for $1 \leq i \leq t-1$ and $z_{s_{\ell}}$ is placed next to $x_{s_{1}}$, so there exists a cycle $x_{s_{1}}, y_{s_{1}}, z_{s_{1}}, x_{s_{2}}, y_{s_{2}}, z_{s_{2}}, \dots, x_{s_{t}}, y_{s_{t}}, z_{s_{t}}, x_{s_{1}}$ in the seat graph.
Since the seat graph is a $2 \times m$ grid, there must be two consecutive bends as depicted in Cases (i), (ii) and (iii) of Figure \ref{fig:8}. 
%Case (i) is impossible because $c_{s_{k}}$ must be a neighbor of $z_{s_{k}}$, but there is no room for it.
First, consider Case (i).
Recall that $c_{s_k}$ (resp.~$c_{s_{k-1}}$) must be isolated or placed next to $z_{s_k}$ (resp.~$z_{s_{k-1}}$), but since there is no available vertex next to $z_{s_k}$ or $z_{s_{k-1}}$, both of $c_{s_k}$ and $c_{s_{k-1}}$ must be isolated.
However, this is a contradiction because at most one agent can be isolated as we have seen before.
Hence, Case (i) is impossible.
For the same reason, neither Case (ii) nor (iii) is possible. Hence we can exclude the existence of a cycle and conclude that $G'$ consists of one path $u_{s_{1}}, u_{s_{2}}, \ldots, u_{s_{n}}$.
By construction of $G'$, for each $i (1\leq i \leq n-1$), $z_{s_{i}}$ is placed next to $x_{s_{i+1}}$ in $\pi$. Thus, by the aforementioned property, there is an arc $(v_{s_{i}}, v_{s_{i+1}})$ in $G$ and hence there is a Hamiltonian path $v_{s_{1}}, v_{s_{2}}, \ldots, v_{s_{n}}$. 
This completes the proof.

\begin{figure}[htbp]
\begin{center}
\includegraphics[width=120mm]{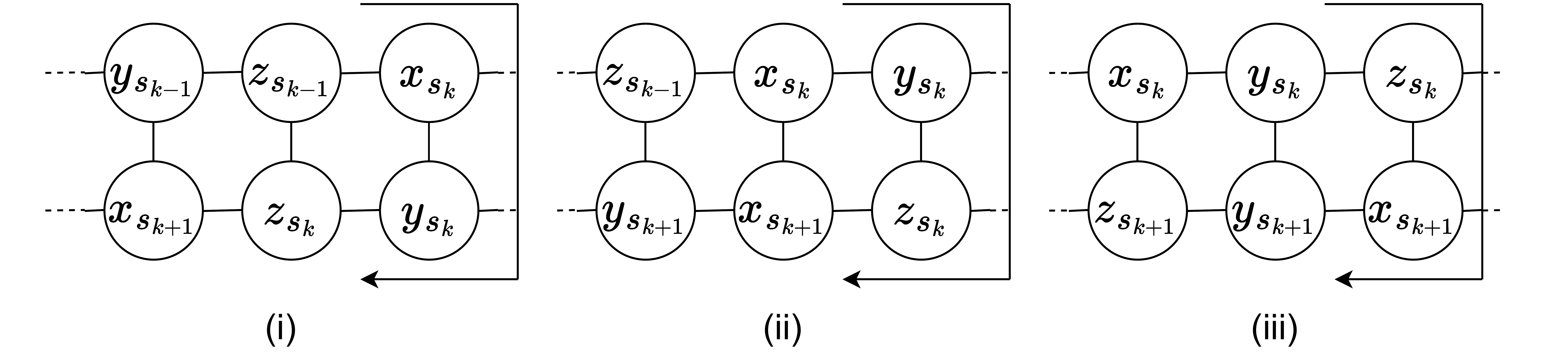}
\caption{Three possible configurations of the seat arrangement $\pi$ when there is a cycle in $G'$}\label{fig:8}
\end{center}
\end{figure}

\end{proof}

\begin{theorem}
Deciding whether an exchange-stable arrangement on a $3 \times m$ grid graph exists is NP-complete. 
\end{theorem}

\begin{proof}
Membership in NP is obvious.
We prove the hardness by a reduction from DHP*. Let $G = (V, E)$ be an instance of DHP* and let $n = |V|$. For each vertex $v_i \in V$, we introduce nine agents $x_i$, $y_i$, $z_i$, $a_i$, $b_i$, $c_i$, $d_i$, $e_i$, and $f_i$. Moreover, we add five agents $s$, $t_1$, $t_2$, $t_3$, $t_4$, and a {\it dummy} agent $D_1$. 
Therefore, there are $9n+6$ agents in total.
Their preference profile $P$ is given in Table \ref{table:4}.
Finally, we set $m = 3n+2$, i.e., our seat graph is a $3\times (3n+2)$ grid graph.

\begin{table}[t]
    
    \centering
    
    \begin{tabular}{c|c|c|c|c}\hline\hline
       & $-17$ & $-1$ & $0$ & $4$ \\ \hline
      $x_i(1 \leq i \leq n)$ & $s$ & other agents & $t_1$, $t_2$, $t_3$, $t_4$ & $y_i$ \\ \hline
      $y_i(1 \leq i \leq n)$ & $s$ & other agents & $t_1$, $t_2$, $t_3$, $t_4$ & $z_i$ \\ \hline
      $z_i(1 \leq i \leq n-1)$ & $s$ & other agents & $t_1$, $t_2$, $t_3$, $t_4$ & $x_p$ $((v_i, v_p)\in E)$ \\ \hline
      $z_n$ & $s$ & other agents & $t_1$, $t_2$, $t_3$, $t_4$ & $c_n$  \\ \hline
      
    \end{tabular}

  \vspace{3mm}

    \centering
    \begin{tabular}{c|c|c|c|c}\hline\hline
       & $-17$ & $-1$ & $0$ & $4$ \\ \hline
      $a_i(1 \leq i \leq n)$ & $s$ & other agents & $t_1$, $t_2$, $t_3$, $t_4$ & $x_i$ \\ \hline
      $b_i(1 \leq i \leq n)$ & $s$ & other agents & $t_1$, $t_2$, $t_3$, $t_4$ & $y_i$ \\ \hline
      $c_i(1 \leq i \leq n)$ & $s$ & other agents & $t_1$, $t_2$, $t_3$, $t_4$ & $z_i$ \\ \hline
      $d_i(1 \leq i \leq n)$ & $s$ & other agents & $t_1$, $t_2$, $t_3$, $t_4$ & $x_i$ \\ \hline
      $e_i(1 \leq i \leq n)$ & $s$ & other agents & $t_1$, $t_2$, $t_3$, $t_4$ & $y_i$ \\ \hline
      $f_i(1 \leq i \leq n)$ & $s$ & other agents & $t_1$, $t_2$, $t_3$, $t_4$ & $z_i$ \\ \hline
    \end{tabular}

\vspace{3mm}

    \centering
    \begin{tabular}{c|c|c}\hline\hline
       & $0$ & $1$  \\ \hline
      $s$ & $t_1$, $t_2$, $t_3$, $t_4$ & other agents   \\ \hline
      $t_i(1 \leq i \leq 4)$ & other agents & $s$ \\ \hline
    \end{tabular}

  \vspace{3mm}

    \centering
    \begin{tabular}{c|c|c}\hline\hline
       & $-17$ & $0$  \\ \hline
      $D_1$ & $s$ & other agents\\ \hline
    \end{tabular}
    \caption{Preference profile $P$}\label{table:4}

\end{table}

We first argue that if there is a Hamiltonian path in $G$, $P$ has an exchange-stable arrangement on a $3\times m$ grid graph. 
Suppose that $G$ has a Hamiltonian path $H$. Recall that $v_n$ has no outgoing edge, so $v_n$ must be the last vertex of $H$. By renaming other vertices, we assume without loss of generality that $H$ is ordered according to the indices of vertices, i.e., $H:= v_1, v_2, \ldots, v_n$.
Corresponding to $H$, define the seat arrangement $\pi$ as shown in Figure \ref{fig:9}, where each agent's utility in $\pi$ is given near the vertex he is seated.

\begin{figure}[htbp]
\begin{center}
\includegraphics[width=150mm]{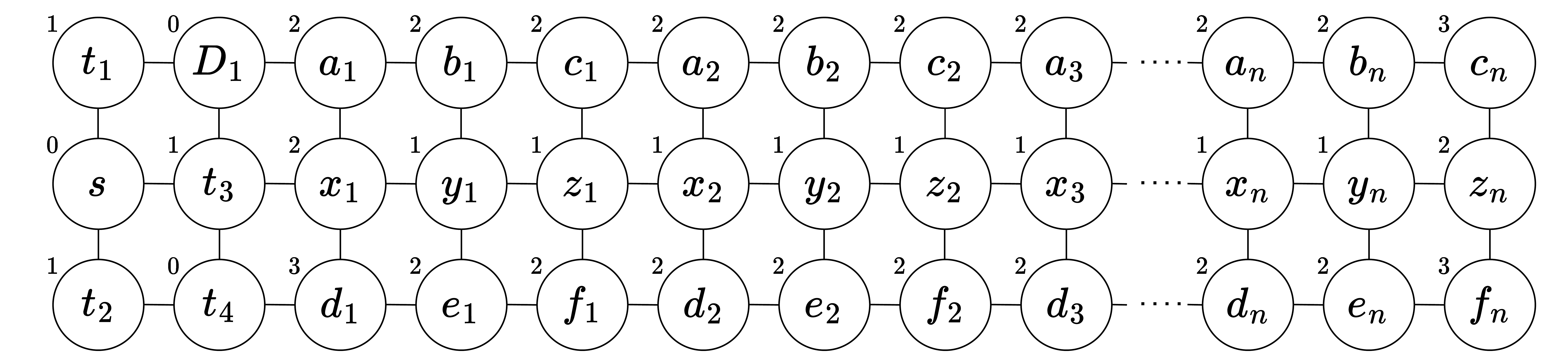}
\caption{The seat arrangement $\pi$}\label{fig:9}
\end{center}
\end{figure}

We will show that $\pi$ is an exchange-stable arrangement. To facilitate the proof, we partition the agents into four groups, and for each group we show that no one is a blocking agent.
For an agent $q \in \{ x_i, y_i, z_i, a_i, b_i, c_i, d_i, e_i, f_i \mid 1 \leq i \leq n \}$, $q$'s {\em favorite} is an agent to whom $q$'s utility is 4. 
For example, $a_3$'s favorite is $x_3$ and $y_5$'s favorite is $z_5$. Note that $z_i (1 \leq i \leq n-1)$ may have more than one favorite. 

{\bf\boldmath Group 1. $s$, $t_1$, $t_2$, $t_3$, $t_4$, and $D_1$.}
First, $D_1$ has no envy because he already has maximum possible utility in $\pi$.
Therefore, he cannot be a blocking agent.
Similarly, $t_1$, $t_2$, and $t_3$ also already have maximum possible utility in $\pi$. 
Hence they have no envy in $\pi$ and cannot be blocking agents.
Agent $t_4$ has utility 0 now, and he can get positive utility only by sitting next to $s$.
Hence, he envies only $t_1$, $t_2$, and $t_3$.
However, none of $t_1$, $t_2$, and $t_3$ has envy in $\pi$ as mentioned above, so $t_4$ cannot be a blocking agent.
Agent $s$'s neighbors are $t_1$, $t_2$, and $t_3$, but utilities of agents $x_i$, $y_i$, $z_i$, $a_i$, $b_i$, $c_i$, $d_i$, $e_i$, and $f_i$ $(1 \leq i \leq n)$ for $t_1$, $t_2$, and $t_3$ are 0, so their utilities decrease to 0 by moving to $\pi (s)$, that is, they do not envy $s$.
Hence, $s$ is not a blocking agent.
Thus no one in Group 1 is a blocking agent.

{\bf\boldmath Group 2. $a_i$, $b_i$, $c_i$, $d_i$, $e_i$, and $f_i$ $(1 \leq i \leq n)$.}
Note that these agents have only one favorite and are now sitting next to the favorite.
Therefore, to increase the utility, an agent still has to sit next to his favorite after the swap.
For instance, $b_2$'s favorite is $y_2$, so $b_2$ may envy only $x_2$, $y_2$, $z_2$, and $e_2$. 
However, in this case, $b_2$ actually has no envy because his utility remains 2 or decreases to 1 by swapping a seat with any one of these four agents.
Checking in this way, we can see that Group 2 agents having 
indices 2 through $n-1$ have no envy in $\pi$, and envy emanating from Group 2 agents is only from $a_1$ to $d_1$.
However, it is easy to see that $d_1$ has no envy to $a_1$.
Therefore, we can conclude that no one in Group 2 can be a blocking agent.

{\bf\boldmath Group 3. $x_i$ and $y_i$ $(1 \leq i \leq n)$.}
Note that these agents have only one favorite and are now sitting next to the favorite.
Therefore, to increase the utility, an agent still has to sit next to his favorite after the swap.
For instance, $y_2$'s favorite is $z_2$, so $y_2$ may envy only $c_2$, $f_2$, $z_2$, and $x_3$. 
However, if $y_2$ swaps a seat with $z_2$ or $x_3$, $y_2$'s utility remains unchanged.
Also, from the above observations for Group 2 agents, neither $c_2$ nor $f_2$ is a blocking agent.
Hence, $y_2$ cannot be a blocking agent.
Checking in this way, we can see that Group 3 agents having 
indices 1 through $n-1$ are not blocking agents.
It remains to consider $x_n$ and $y_n$.
First, $x_n$ envies $b_n$, $e_n$, and $z_n$.
However, $b_n$ and $e_n$ are Group 2 agents, so we have seen that they are not blocking agents.
It is easy to see that $z_n$ does not envy $x_n$.
Hence, $x_n$ is not a blocking agent.
Next, $y_n$ envies $c_n$, $f_n$, and $z_n$.
However, $c_n$ and $f_n$ are Group 2 agents, so they are not blocking agents, and it is easy to see that $z_n$ does not envy $y_n$.
Hence, $y_n$ is not a blocking agent and we can conclude that no one in Group 3 can be a blocking agent.

{\bf\boldmath Group 4. $z_i$ $(1 \leq i \leq n)$.}
We now have only to care about a blocking pair within Group 4 agents.
Observe that for any $z_i$ and $z_j$, $z_i$ does not envy $z_j$ because $z_i$ cannot increase his utility by moving to $\pi(z_j)$.
Thus no one in Group 4 is a blocking agent.

From the above discussion, there is no blocking pair in $\pi$, hence $\pi$ is exchange-stable.

For the opposite direction, assume that there is an exchange-stable arrangement $\pi$ of $P$ on a $3 \times m$ grid graph.
We will show that $G$ has a Hamiltonian path.
We say that an agent $p$ is {\it isolated} if $p$'s neighbors in $\pi$ are included in $\{t_1, t_2, t_3, t_4\}$.
We first show that $s$ is isolated.
Suppose not, and let $p \not \in \{t_1, t_2, t_3, t_4\}$ be an agent who is seated next to $s$.
Since the maximum degree of the seat graph is four, there must be an agent $t \in \{t_1, t_2, t_3, t_4\}$ who is {\it not} seated next to $s$.
Note that $p$'s utility is now at most $-17+4\times 3=-5$, but if he moves to $\pi(t)$, the utility becomes at least $-4$.
Similarly, $t$'s utility is 0 now, but if he moves to $\pi(p)$, it increase to 1.
These mean that $(p, t)$ is a blocking pair, contradicting the stability of $\pi$.
Therefore, $s$ is isolated.
Note that there can be at most two isolated agents in $A \setminus \{ s \}$, as shown in Figure {\ref{fig:10}}, where in (i) two agents $b_3$ and $c_4$ are isolated, while in (ii) one agent $a_6$ is isolated.

\begin{figure}[htbp]
\begin{center}
\includegraphics[width=110mm]{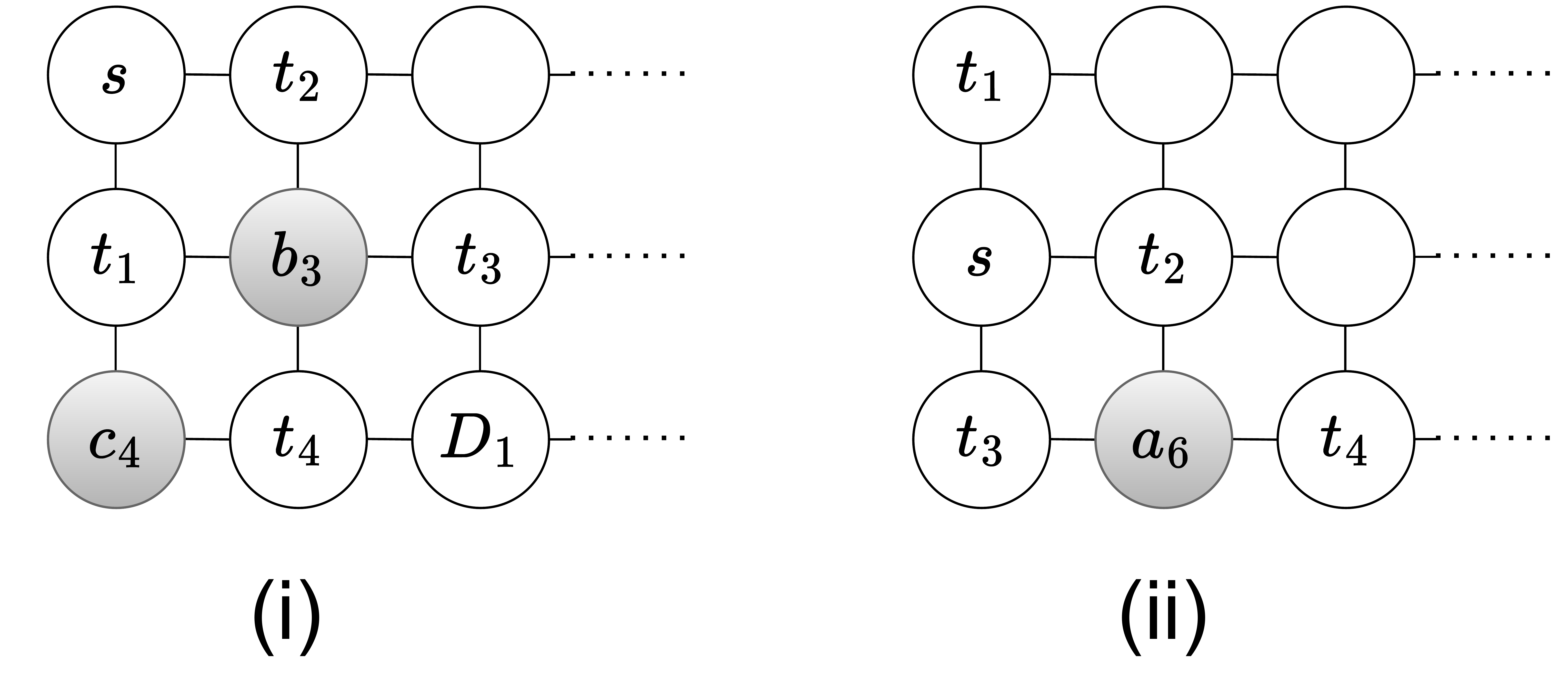}
\caption{Examples of $\pi$ including one or two isolated agents in $A \setminus \{ s \}$}\label{fig:10}
\end{center}
\end{figure}

Next, we claim that $x_i (1 \leq i \leq n)$ is not isolated for the following reason:
If $x_i$ is isolated, then neither $a_i$ nor $d_i$ is a neighbor of $x_i$.
Since there can be at most two isolated agents other than $s$, $a_i$ or $d_i$ is not isolated.
Suppose that $a_i$ is not isolated (the same argument below holds in case $d_i$ is not isolated). 
Then $a_i$ has at least one neighbor to whom his utility is $-1$, so his utility is negative, but he can increase the utility to 0 by moving to $\pi(s)$.
On the other hand, $s$'s utility is 0, but he can have a positive utility by moving to $\pi(a_i)$ because $a_i$ is not isolated.
Therefore $(a_i, s)$ is a blocking pair, contradicting the stability of $\pi$.
Hence, $x_i$ is not isolated.
For the same reason, we can see that none of $y_i (1 \leq i \leq n)$, $z_i (1 \leq i \leq n)$ are isolated.

Moreover, we claim that if $a_i$ (resp.~$d_i$) $(1 \leq i \leq n)$ is not isolated, $a_i$ (resp.~$d_i$) is placed next to $x_i$.  For, if not, by swapping $a_i$ (resp.~$d_i$) and $s$, $a_i$ (resp.~$d_i$) can increase the utility from at most $-1$ to $0$ and $s$ can increase the utility from $0$ to at least $1$. Therefore, $(a_i, s)$ (resp.~$(d_i, s)$) is a blocking pair, a contradiction. 
For the same reason, for $1 \leq i \leq n$, $b_i$ and $e_i$ must be placed next to $y_i$ if not isolated, and $c_i$ and $f_i$ must be placed next to $z_i$ if not isolated. 
Also, for the same reason, we can see that $x_i$ is placed next to $y_i$ and $y_i$ is placed next to $z_i$ for $1 \leq i \leq n$.

Then, let us consider $z_i (1 \leq i \leq n-1)$. 
Recall that the vertex $v_i$ has at least one outgoing edge in $G$.
We will show that $z_i$ is placed next to some $x_p$ such that $(v_i, v_p) \in E$.
Suppose not. Then $z_i$'s utility is at most $-1$ because $z_i$ is not isolated as proved above. 
Then, by swapping $z_i$ and $s$, $z_i$ can increase the utility from at most $-1$ to $0$ and $s$ can increase the utility from $0$ to at least $1$. Therefore, $(z_i, s)$ is a blocking pair, a contradiction. 
We also show that $z_{i}$ can never have more than one $x_{p}$s as his neighbor.
Assume on the contrary that $x_{p}$ and $x_q$ are $z_i$'s neighbors. 
Recall that $y_i$ must also be a neighbor of $z_i$, so $z_i$ is placed at a vertex of degree 3 or 4. 
First, suppose that the degree of $\pi(z_i)$ is 4 (Figure \ref{fig:11}(i)).
Recall that $x_i$ must be a neighbor of $y_i$, and $b_i$ and $e_i$ must be neighbors of $y_i$ if they are not isolated.
Similarly, $y_p$ (resp.~$y_q$) is a neighbor of $x_p$ (resp.~$x_q$) and $a_p$ and $d_p$ (resp.~$a_q$ and $d_q$), if not isolated, must be neighbors of $x_p$ (resp.~$x_q$).
To sum up, each of these nine agents must either be placed at one of six vertices 1 through 6 in Figure \ref{fig:11}(i) or be isolated.
However, as we have seen above, at most two agents can be isolated, so this is impossible.
There are several more cases according whether how $y_i$, $x_p$, and $x_q$ are placed around $z_i$, but it is easy to see that any one of them is impossible by the same reasoning.
The case when the degree of $\pi(z_i)$ is 3 (Figure \ref{fig:11}(ii)) can be argued similarly.
In this case, each of these nine agents must be placed at one of five vertices 1 through 5 or must be isolated, which is impossible.

Next, we show that $x_i (1 \leq i \leq n)$ can have at most one $z_p$ as a neighbor.
This can be shown in the same manner as for $z_i$ above, by assuming that $z_p$ and $z_q$ are neighbors of $x_i$ and by noting that
\begin{itemize}
\item $y_i$ must be a neighbor of $x_i$,

\item $b_i$, $e_i$ and $z_i$ must be neighbors of $y_i$,

\item $y_p$, $c_p$ and $f_p$ must be neighbors of $z_p$, and

\item $y_q$, $c_q$ and $f_q$ must be neighbors of $z_q$.

\end{itemize}

\begin{figure}[htbp]
\begin{center}
\includegraphics[width=140mm]{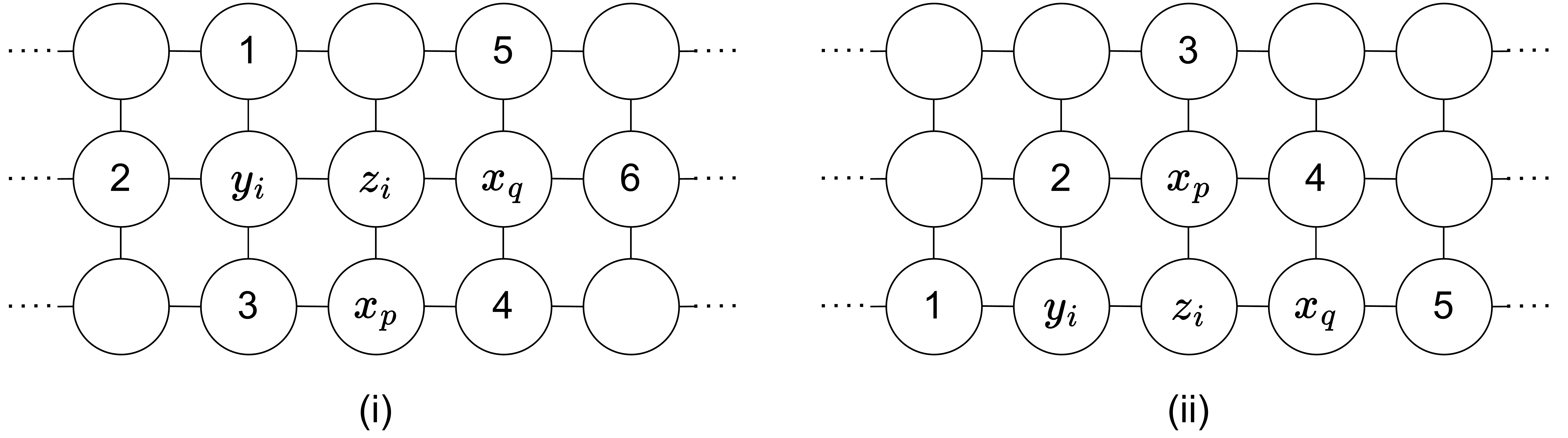}
\caption{Possible configurations of $y_j$, $x_p$, and $x_q$}\label{fig:11}
\end{center}
\end{figure}

Now, from $\pi$, construct the directed graph $G'=(V', E')$ where $V'= \{ u_1, u_2, \ldots, u_n \}$ and there is an arc $(u_{i}, u_{j}) \in E'$ if and only if $z_i$ is placed next to $x_j$ in $\pi$. From the above observations, we know that each $u_{i} (1 \leq i \leq n-1)$ has exactly one outgoing edge and each $u_{i} (1 \leq i \leq n)$ has at most one incoming edge. These facts imply that $G'$ consists of at most one directed path and some number of directed cycles. In the following, we will show that there is no cycle. 
Assume on the contrary that there exists a cycle $u_{s_{1}}, u_{s_{2}}, \ldots, u_{s_{t}}, u_{s_{1}}$ of length $t$.
Then, in the seat arrangement $\pi$, $z_{s_{i}}$ is placed next to $x_{s_{i+1}}$ for $1 \leq i \leq t-1$ and $z_{s_{t}}$ is placed next to $x_{s_{1}}$, so there exists a cycle $x_{s_{1}}, y_{s_{1}}, z_{s_{1}}, x_{s_{2}}, y_{s_{2}}, z_{s_{2}}, \dots, x_{s_{t}}, y_{s_{t}}, z_{s_{t}}, x_{s_{1}}$ in the seat graph.
Now, define the {\em leftmost vertices} as the vertices on this cycle lying on the leftmost column of the seat graph.
Since the graph has three rows, there are two or three leftmost vertices, as depicted in (i), (ii), and (iii) of Figure {\ref{fig:12}}.
Consider Figure {\ref{fig:12}} (i), where the leftmost vertices are $\pi(y_{s_i})$ and $\pi(z_{s_i})$ (here, we fixed agents to $y_{s_i}$ and $z_{s_{i}}$, but it is easy to see that the following argument holds for any choice).
Since this is a part of a cycle, $y_{s_i}$ and $z_{s_{i}}$'s right-hand neighbors can be determined to $x_{s_i}$ and $x_{s_{i+1}}$, respectively.
Furthermore, we can see that $x_{s_i}$'s right-hand neighbors is $z_{s_{i-1}}$.
Recall that $b_{s_i}$ and $e_{s_i}$ must be isolated or placed next to $y_{s_i}$, but since there is only one available vertex next to $y_{s_i}$, at least one of them must be isolated.
Similarly, $a_{s_i}$ and $d_{s_i}$ must be isolated or placed next to $x_{s_i}$, but there is no available vertex next to $x_{s_i}$, so both of them must be isolated.
However, as we have seen before, at most two agents can be isolated, a contradiction.
By symmetry, the case of Figure {\ref{fig:12}} (ii) can be handled similarly.
Finally, consider the case of Figure {\ref{fig:12}} (iii). 
Similarly as above, $x_{s_i}$ and $z_{s_i}$'s right-hand neighbors can be determined to $z_{s_{i-1}}$ and $x_{s_{i+1}}$, respectively.
Recall that $a_{s_i}$ and $d_{s_i}$ must be isolated or placed next to $x_{s_i}$, but since there is only one available vertex next to $x_{s_i}$, at least one of them must be isolated.
Similarly, $c_{s_i}$ and $f_{s_i}$ must be isolated or placed next to $z_{s_i}$, but since there is only one available vertex next to $z_{s_i}$, at least one of them must be isolated.
For the same reason, $c_{s_{i-1}}$ and $f_{s_{i-1}}$ must be isolated or placed at vertices 1 and 2, and $a_{s_{i+1}}$ and $d_{s_{i+1}}$ must be isolated or placed at vertices 2 and 3.
Hence at least one of these four agents must be isolated.
From the above arguments, at least three agents must be isolated, but this is impossible.
Thus we have excluded the existence of a cycle and can conclude that $G'$ consists of one path $u_{s_{1}}, u_{s_{2}}, \ldots, u_{s_{n}}$.
By construction of $G'$, for each $i (1\leq i \leq n-1$), $z_{s_{i}}$ is placed next to $x_{s_{i+1}}$ in $\pi$. Thus, by the aforementioned property, there is an arc $(v_{s_{i}}, v_{s_{i+1}})$ in $G$ and hence there is a Hamiltonian path $v_{s_{1}}, v_{s_{2}}, \ldots, v_{s_{n}}$. 
This completes the proof.

\begin{figure}[htbp]
\begin{center}
\includegraphics[width=140mm]{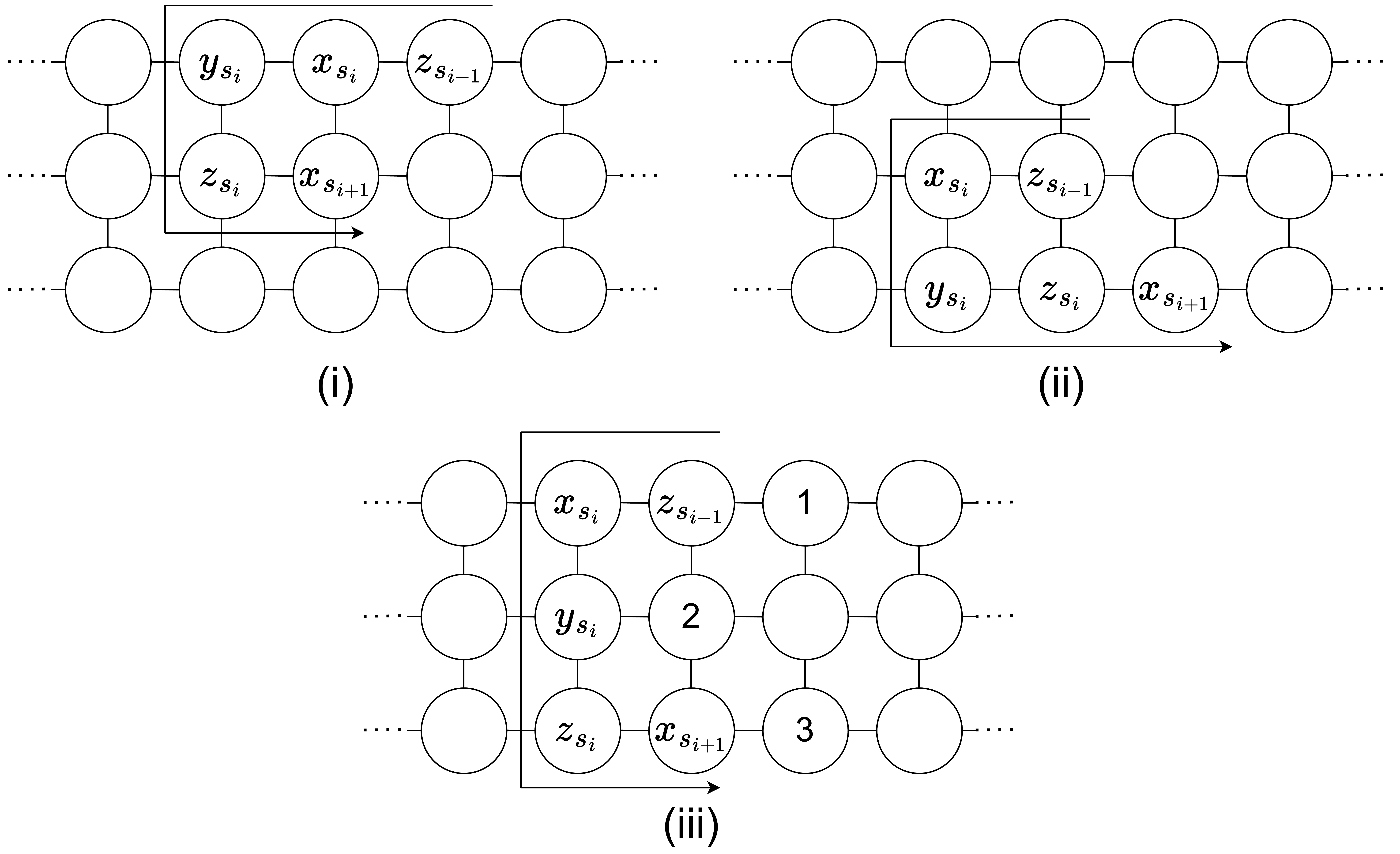}
\caption{Three possible configurations of the seat arrangement $\pi$ when there is a cycle in $G'$}\label{fig:12}
\end{center}
\end{figure}

\end{proof}

\begin{theorem}
Deciding whether an exchange-stable arrangement on an $\ell \times m$ grid graph exists is NP-complete for $\ell \geq 4$.
\end{theorem}

\begin{proof}
Membership in NP is obvious.
We prove the hardness by a reduction from DHP*. Let $G = (V, E)$ be an instance of DHP* and let $n = |V|$. For each vertex $v_i \in V$, we introduce nine agents $x_i$, $y_i$, $z_i$, $a_i$, $b_i$, $c_i$, $d_i$, $e_i$, and $f_i$. Moreover, we add five agents $s$, $t_1$, $t_2$, $t_3$, $t_4$, and {\it dummy} agents $D_i (1\leq i \leq 3n\ell+2\ell-9n-5)$.
Therefore, there are $3n\ell+2\ell$ agents in total.
Their preference profile $P$ is given in Table {\ref{table:5}}.
Finally, we set $m=3n+2$, i.e., our seat graph is an $\ell \times (3n+2)$ grid graph.

\begin{table}[t]
    \centering
     
    \begin{tabular}{c|c|c|c|c}\hline\hline
       & $-17$ & $-1$ & $0$ & $4$ \\ \hline
      $x_i(1 \leq i \leq n)$ & $s$ & other agents & $t_1$, $t_2$, $t_3$, $t_4$ & $y_i$ \\ \hline
      $y_i(1 \leq i \leq n)$ & $s$ & other agents & $t_1$, $t_2$, $t_3$, $t_4$ & $z_i$ \\ \hline
      $z_i(1 \leq i \leq n-1)$ & $s$ & other agents & $t_1$, $t_2$, $t_3$, $t_4$ & $x_p$ $((v_i, v_p)\in E)$ \\ \hline
      $z_n$ & $s$ & other agents & $t_1$, $t_2$, $t_3$, $t_4$ & $c_n$  \\ \hline
      
    \end{tabular}

  \vspace{3mm}

    \centering
    \begin{tabular}{c|c|c|c|c}\hline\hline
       & $-17$ & $-1$ & $0$ & $4$ \\ \hline
      $a_i(1 \leq i \leq n)$ & $s$ & other agents & $t_1$, $t_2$, $t_3$, $t_4$ & $x_i$ \\ \hline
      $b_i(1 \leq i \leq n)$ & $s$ & other agents & $t_1$, $t_2$, $t_3$, $t_4$ & $y_i$ \\ \hline
      $c_i(1 \leq i \leq n)$ & $s$ & other agents & $t_1$, $t_2$, $t_3$, $t_4$ & $z_i$ \\ \hline
      $d_i(1 \leq i \leq n)$ & $s$ & other agents & $t_1$, $t_2$, $t_3$, $t_4$ & $x_i$ \\ \hline
      $e_i(1 \leq i \leq n)$ & $s$ & other agents & $t_1$, $t_2$, $t_3$, $t_4$ & $y_i$ \\ \hline
      $f_i(1 \leq i \leq n)$ & $s$ & other agents & $t_1$, $t_2$, $t_3$, $t_4$ & $z_i$ \\ \hline
    \end{tabular}

   \vspace{3mm}

    \centering
    \begin{tabular}{c|c|c}\hline\hline
       & $0$ & $1$  \\ \hline
      $s$ & $t_1$, $t_2$, $t_3$, $t_4$ & other agents   \\ \hline
      $t_i(1 \leq i \leq 4)$ & other agents & $s$ \\ \hline
      
    \end{tabular}

  \vspace{3mm}

    \centering
    \begin{tabular}{c|c|c}\hline\hline
       & $-17$ & $0$  \\ \hline
      $D_i(1 \leq i \leq 3n\ell+2\ell-9n-5)$ & $s$ & other agents   \\ \hline

    \end{tabular}
    \caption{Preference profile $P$}\label{table:5}
 
\end{table}

We first argue that if there is a Hamiltonian path in $G$, $P$ has an exchange-stable arrangement on an $\ell\times m$ grid graph. 
Suppose that $G$ has a Hamiltonian path $H$.
Recall that $v_n$ has no outgoing edge, so $v_n$ must be the last vertex of $H$.
By renaming other vertices, we assume without loss of generality that $H$ is ordered according to the indices of vertices, i.e., $H:= v_1, v_2, \ldots, v_n$.
Corresponding to $H$, define the seat arrangement $\pi$ as shown in Figure \ref{fig:13}, where each agent's utility in $\pi$ is given near the vertex he is seated.

\begin{figure}[htbp]
\begin{center}
\includegraphics[width=150mm]{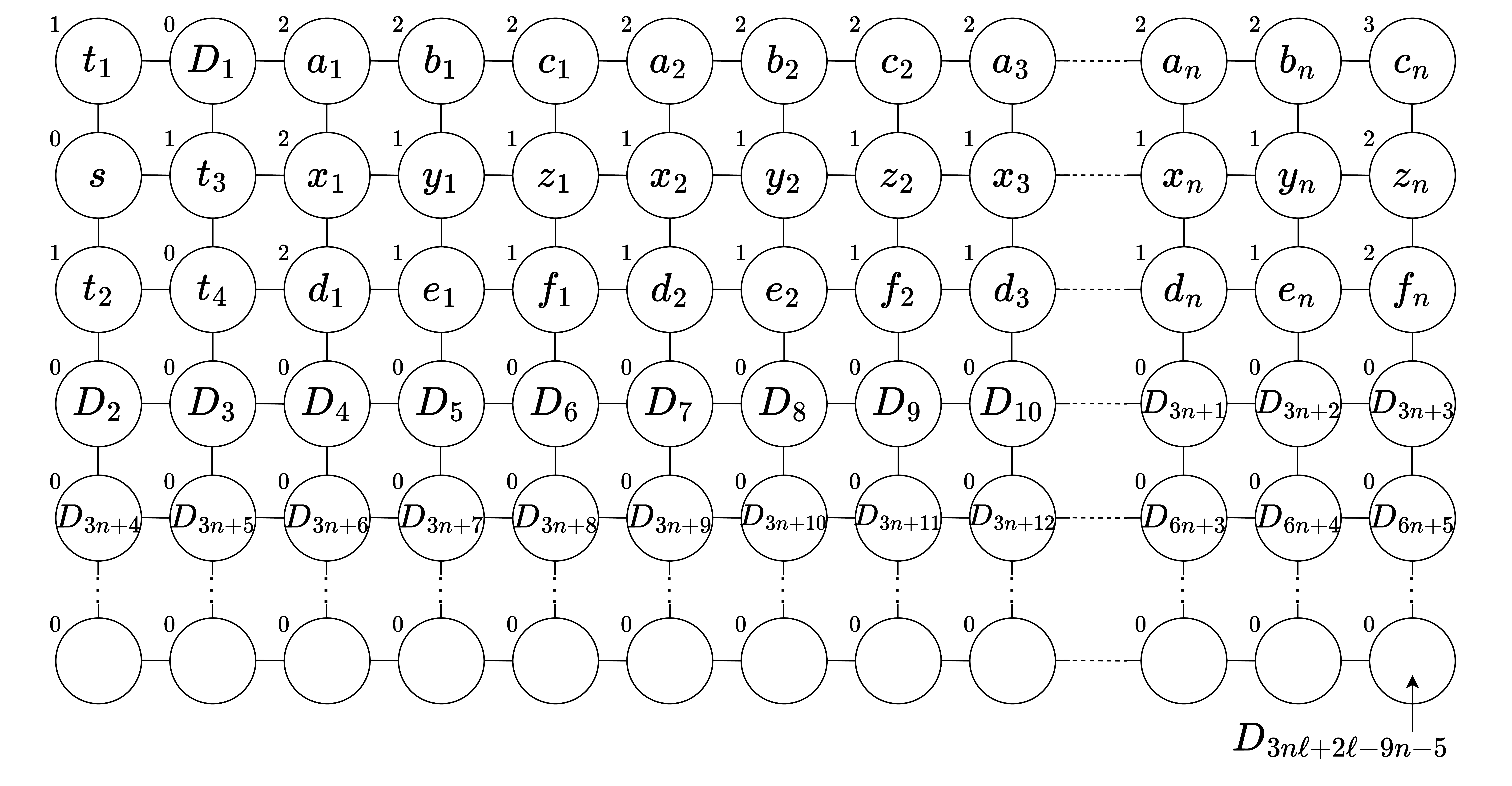}
\caption{The seat arrangement $\pi$}\label{fig:13}
\end{center}
\end{figure}

We will show that $\pi$ is an exchange-stable arrangement. To facilitate the proof, we partition the agents into four groups, and for each group we show that no one is a blocking agent.
For an agent $q \in \{ x_i, y_i, z_i, a_i, b_i, c_i, d_i, e_i, f_i \mid 1 \leq i \leq n \}$, $q$'s {\em favorite} is an agent to whom $q$'s utility is 4. 
For example, $b_3$'s favorite is $y_3$ and $y_2$'s favorite is $z_2$. Note that $z_i (1 \leq i \leq n-1)$ may have more than one favorite. 

{\bf\boldmath Group 1. $s$, $t_1$, $t_2$, $t_3$, $t_4$, and $D_i (1\leq i \leq 3n\ell+2\ell-9n-5)$.}
First, for any $i$, $D_i$ has no envy because he already has maximum possible utility in $\pi$.
Therefore, $D_i$ cannot be a blocking agent. 
Similarly, $t_1$, $t_2$, and $t_3$ also already have maximum possible utility in $\pi$. 
Hence they have no envy in $\pi$ and so cannot be blocking agents.
Agent $t_4$ has utility 0 now, and he can get positive utility only by sitting next to $s$.
Hence, he envies only $t_1$, $t_2$, and $t_3$.
However, none of $t_1$, $t_2$, and $t_3$ have envy in $\pi$ as mentioned above, so $t_4$ cannot be a blocking agent. 
Agent $s$'s neighbors are $t_1$, $t_2$, and $t_3$, but utilities of agents $x_i$, $y_i$, $z_i$, $a_i$, $b_i$, $c_i$, $d_i$, $e_i$, and $f_i$ $(1 \leq i \leq n)$ for $t_1$, $t_2$, and $t_3$ are 0, so their utilities decrease to 0 by moving to $\pi (s)$, that is, they do not envy $s$.
Hence, $s$ cannot be a blocking agent.
Thus no agent in Group 1 is a blocking agent. 

{\bf\boldmath Group 2. $a_i$, $b_i$, and $c_i$ $(1 \leq i \leq n)$.}
Note that these agents have only one favorite and are now sitting next to the favorite.
Therefore, to increase the utility, an agent still has to sit next to his favorite after the swap.
For instance, $b_2$'s favorite is $y_2$, so $b_2$ may envy only $x_2$, $y_2$, $z_2$, and $e_2$. 
However, in this case, $b_2$ actually has no envy because his utility decreases to 1 by swapping a seat with any one of these four agents.
Checking in this way, we can see that Group 2 agents have no envy in $\pi$.
Therefore, we can conclude that no agent in Group 2 is a blocking agent.

{\bf\boldmath Group 3. $x_i$, $y_i$, $d_i$, $e_i$, and $f_i$ $(1 \leq i \leq n)$.}
Note that these agents have only one favorite and are now sitting next to the favorite.
Therefore, to increase the utility, an agent still has to sit next to his favorite after the swap.
For instance, $y_2$'s favorite is $z_2$, so $y_2$ may envy only $c_2$, $f_2$, $z_2$, and $x_3$. 
However, if $y_2$ swaps a seat with $z_2$, $x_3$, or $f_2$, $y_2$'s utility remains unchanged.
Also, from the fact that agents in Group 2 are not blocking agents, $c_2$ is not a blocking agent.
Hence, $y_2$ cannot be a blocking agent.
Checking in this way, we can see that Group 3 agents having 
indices 2 through $n-1$ are not blocking agents.
By a similar argument, we can also see that, except for $e_1$, $x_n$, $y_n$, and $e_n$, Group 3 agents having indices 1 or $n$ are not blocking agents. 
It then remains to consider $e_1$, $x_n$, $y_n$, and $e_n$.
First, $e_1$ envies $b_1$ and $x_1$.
However, $b_1$ is a Group 2 agent, so $b_1$ is not a blocking agent.
Moreover, if $x_1$ swaps a seat with $e_1$, $x_1$'s utility decreases to 1, so $x_1$ does not envy $e_1$.
Secondly, $x_n$ envies $b_n$ and $z_n$.
However, $b_n$ is a Group 2 agent, so $b_n$ is not a blocking agent.
Moreover, if $z_n$ swaps a seat with $x_n$, $z_n$'s utility decreases, so $z_n$ does not envy $x_n$.
Thirdly, $y_n$ envies $c_n$, $z_n$, and $f_n$.
However, $c_n$ is Group 2 agent, so $c_n$ is not a blocking agent.
Moreover, if $z_n$ (resp.~$f_n$) swaps a seat with $y_n$, $z_n$'s (resp.~$f_n$'s) utility decreases, so $z_n$ (resp.~$f_n$) does not envy $y_n$.
Finally, $e_n$ envies $b_n$ and $z_n$.
However, $b_n$ is a Group 2 agent, so $b_n$ is not a blocking agent.
Moreover, if $z_n$ swaps a seat with $e_n$, $z_n$'s utility decreases, so $z_n$ does not envy $x_n$.
Therefore, we can conclude that no agent in Group 3 is a blocking agent.

{\bf\boldmath Group 4. $z_i$ $(1 \leq i \leq n)$.}
We now have only to care about blocking pairs within Group 4 agents.
Observe that for any $z_i$ and $z_j$, $z_i$ does not envy $z_j$ because $z_i$ cannot increase his utility by moving to $\pi(z_j)$.
Thus no agent in Group 4 can be a blocking agent.

From the above discussion, there is no blocking agent in $\pi$, hence $\pi$ is exchange-stable.

For the opposite direction, assume that there is an exchange-stable arrangement $\pi$ of $P$ on an $\ell \times m$ grid graph.
We will show that $G$ has a Hamiltonian path.
We say that an agent $p$ is {\it isolated} if $p$'s neighbors in $\pi$ are included in $\{t_1, t_2, t_3, t_4\}$.
We first show that $s$ is isolated.
Suppose not, and let $p \not \in \{t_1, t_2, t_3, t_4\}$ be an agent who is seated next to $s$.
Since the maximum degree of the seat graph is four, there must be an agent $t \in \{t_1, t_2, t_3, t_4\}$ who is {\it not} seated next to $s$.
Note that $p$'s utility is now at most $-17+4\times 3=-5$, but if he moves to $\pi(t)$, the utility becomes at least $-4$.
Similarly, $t$'s utility is 0 now, but if he moves to $\pi(p)$, it increases to 1.
These mean that $(p, t)$ is a blocking pair, contradicting the stability of $\pi$.
Therefore, $s$ is isolated.
Note that there can be at most one isolated agent in $A \setminus \{ s \}$, as shown in Figure \ref{fig:14}, where $a_8$ is isolated.

\begin{figure}[htbp]
\begin{center}
\includegraphics[width=105mm]{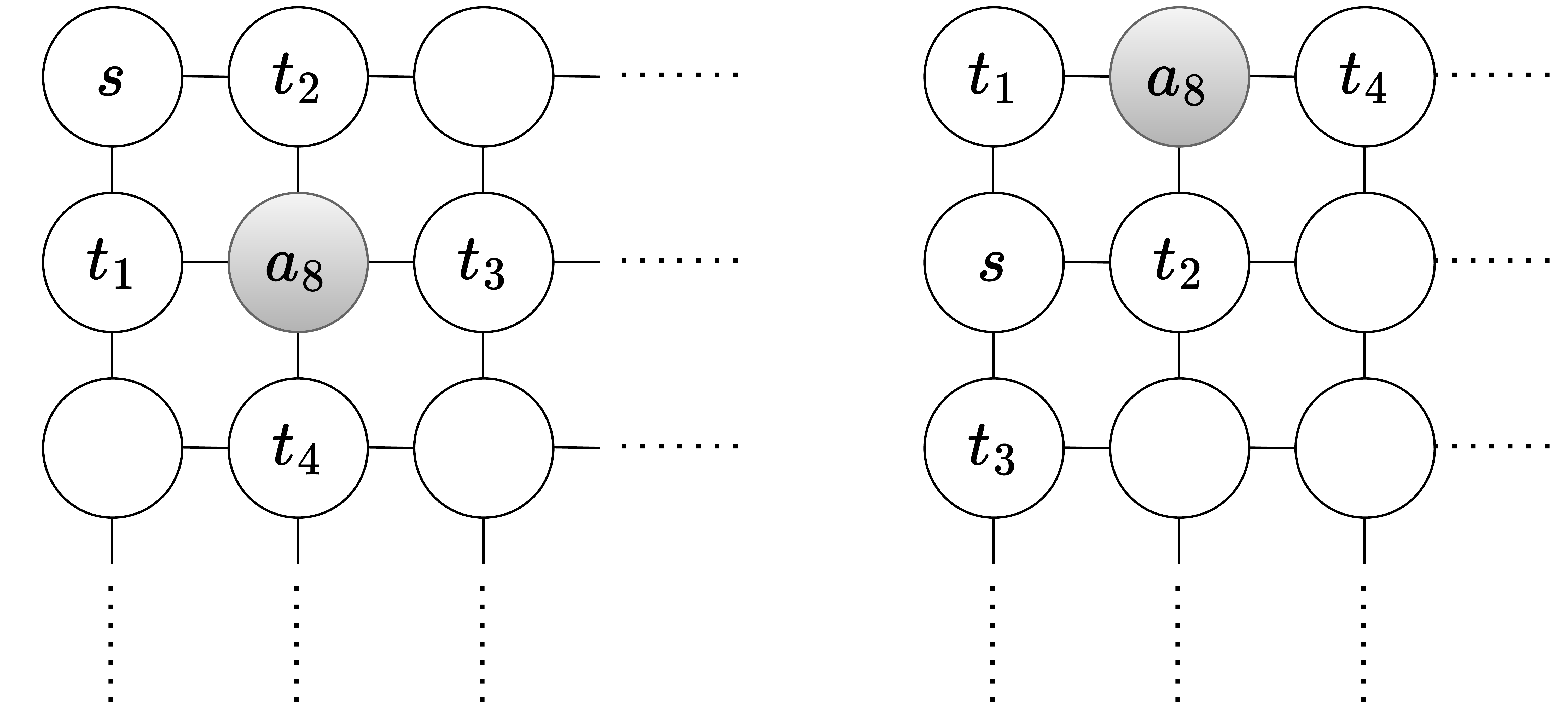}
\caption{Examples of $\pi$ including one isolated agent in $A \setminus \{ s \}$}\label{fig:14}
\end{center}
\end{figure}

Next, we claim that $x_i (1 \leq i \leq n)$ is not isolated for the following reason: 
If $x_i$ is isolated, then $a_i$ is not a neighbor of $x_i$.
Since $a_i$ is not isolated (as there can be at most one isolated agent other than $s$), $a_i$ has at least one neighbor to whom $a_i$'s utility is $-1$.
Thus $a_i$ can increase the utility by moving to $\pi(s)$.
Since $s$ is isolated, his utility is 0, but he can have a positive utility by moving to $\pi(a_i)$ since $a_i$ is not isolated.
Therefore $(a_i, s)$ is a blocking pair, contradicting the stability of $\pi$.
For the same reason, we can see that none of $y_i (1 \leq i \leq n)$, $z_i (1 \leq i \leq n)$, and $c_n$ are isolated.

Moreover, we claim that if $a_i$ (resp.~$d_i$) $(1 \leq i \leq n)$ is not isolated, $a_i$ (resp.~$d_i$) is placed next to $x_i$.  For, if not, by swapping $a_i$ (resp.~$d_i$) and $s$, $a_i$ (resp.~$d_i$) can increase the utility from at most $-1$ to $0$ and $s$ can increase the utility from $0$ to at least $1$. Therefore, $(a_i, s)$ (resp.~$(d_i, s)$) is a blocking pair, a contradiction. 
For the same reason, for $1 \leq i \leq n$, $b_i$ and $e_i$ must be placed next to $y_i$ if not isolated, and $c_i$ and $f_i$ must be placed next to $z_i$ if not isolated. 
Also, for the same reason, we can see that $x_i$ is placed next to $y_i$ and $y_i$ is placed next to $z_i$ in $\pi$ for $1 \leq i \leq n$.

Then, let us consider $z_i (1 \leq i \leq n-1)$. 
Recall that the vertex $v_i$ has at least one outgoing edge in $G$.
We will show that $z_i$ is placed next to some $x_p$ such that $(v_i, v_p) \in E$.
Suppose not. Then $z_i$'s utility is at most $-1$ because $z_i$ is not isolated as proved above. 
Then, by swapping $z_i$ and $s$, $z_i$ can increase the utility from at most $-1$ to $0$ and $s$ can increase the utility from $0$ to at least $1$. Therefore, $(z_i, s)$ is a blocking pair, a contradiction. 
We also show that $z_{i}$ can never have more than one $x_{p}$s as his neighbor.
Assume on the contrary that $x_{p}$ and $x_q$ are $z_i$'s neighbors. 
Recall that $y_i$ must also be a neighbor of $z_i$, so $z_i$ is placed at a vertex of degree 3 or 4.
First, suppose that the degree of $\pi (z_i)$ is 4 (Figure \ref{fig:15}(i)).
Recall that $x_i$ must be a neighbor of $y_i$, and $b_i$ and $e_i$ must be neighbors of $y_i$ if they are not isolated.
Similarly, $y_p$ (resp.~$y_q$) must be a neighbor of $x_p$ (resp.~$x_q$) and $a_p$ and $d_p$ (resp.~$a_q$ and $d_q$), if not isolated, must be neighbors of $x_p$ (resp.~$x_q$).
To sum up, these nine agents must be placed at one of seven vertices 1 through 7 in Figure \ref{fig:15}(i) or must be isolated, so at least two of them must be isolated.
However, this is impossible because at most one agent can be isolated. 
There are several more cases according whether how $y_i$, $x_p$, and $x_q$ are placed around $z_i$, but it is easy to see that any one of them is impossible by the same reasoning.
The case when the degree of $\pi(z_i)$ is 3 (Figure \ref{fig:15}(ii)) can be argued similarly.
In this case, these nine agents must be placed at one of five vertices 1 through 5 or must be isolated, which is impossible.

Next, we show that $x_i (1 \leq i \leq n)$ can have at most one $z_p$ as a neighbor.
This can be shown in the same manner as for $z_i$ above, by assuming that $z_p$ and $z_q$ are neighbors of $x_i$ and by noting that
\begin{itemize}
\item $y_i$ must be a neighbor of $x_i$,

\item $b_i$, $e_i$ and $z_i$ must be neighbors of $y_i$,

\item $y_p$, $c_p$ and $f_p$ must be neighbors of $z_p$, and

\item $y_q$, $c_q$ and $f_q$ must be neighbors of $z_q$.

\end{itemize}

\begin{figure}[htbp]
\begin{center}
\includegraphics[width=135mm]{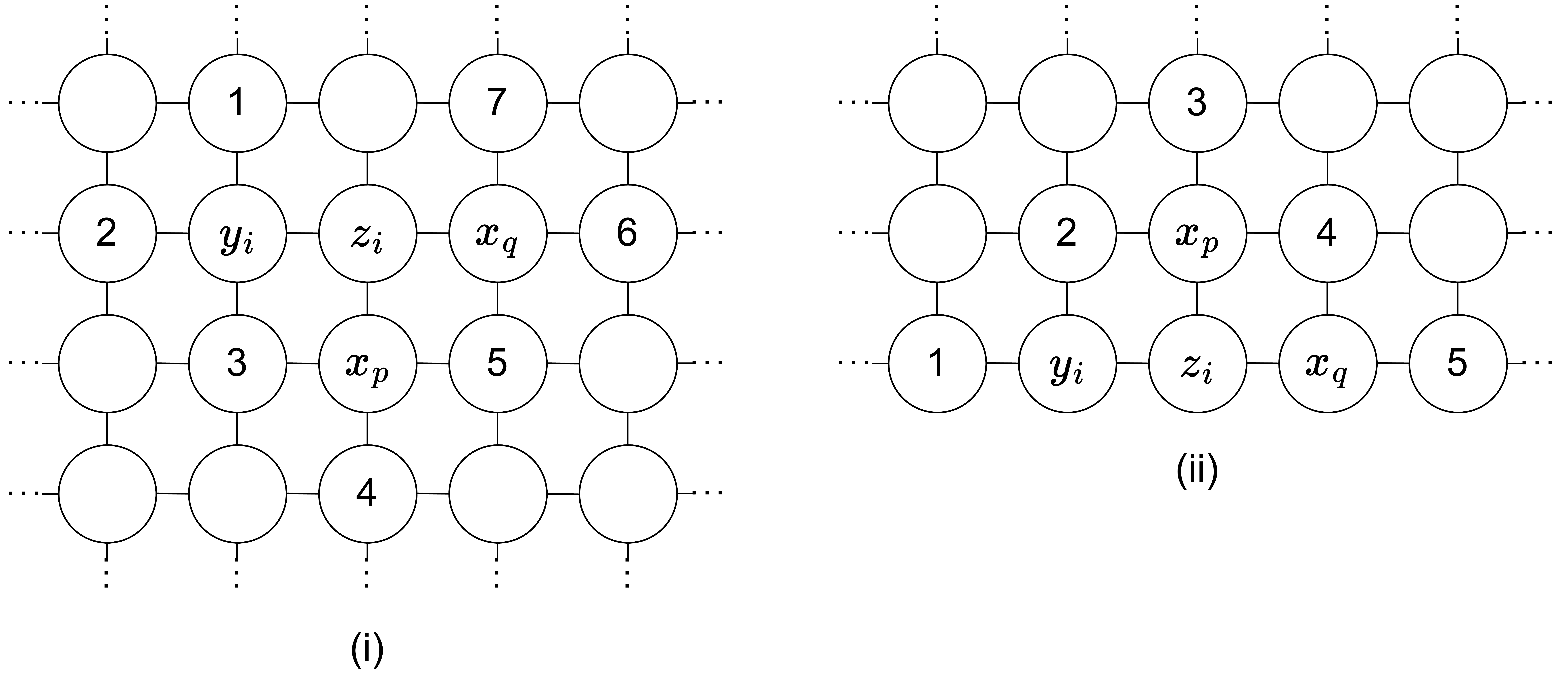}
\caption{One possible configuration of $y_j$, $x_p$, and $x_q$}\label{fig:15}
\end{center}
\end{figure}

Now, from $\pi$, construct the directed graph $G'=(V', E')$ where $V'= \{ u_1, u_2, \ldots, u_n \}$ and there is an arc $(u_{i}, u_{j}) \in E'$ if and only if $z_i$ is placed next to $x_j$ in $\pi$. From the above observations, we know that each $u_{i} (1 \leq i \leq n-1)$ has exactly one outgoing edge and each $u_{i} (1 \leq i \leq n)$ has at most one incoming edge. These facts imply that $G'$ consists of at most one directed path and some number of directed cycles. In the following, we will show that there is no cycle. 
Assume on the contrary that there exists a cycle $u_{s_{1}}, u_{s_{2}}, \ldots, u_{s_t}, u_{s_{1}}$ of length $t$.
Then, in the seat arrangement $\pi$, $z_{s_{i}}$ is placed next to $x_{s_{i+1}}$ for $1 \leq i \leq t-1$ and $z_{s_{t}}$ is placed next to $x_{s_{1}}$, so there exists a cycle $x_{s_{1}}, y_{s_{1}}, z_{s_{1}}, x_{s_{2}}, y_{s_{2}}, z_{s_{2}}, \dots, x_{s_{t}}, y_{s_{t}}, z_{s_{t}}, x_{s_{1}}$ in the seat graph.
This implies that there exists a bend where the cycle changes direction, as depicted Figure \ref{fig:16}(i). (Here, the bend occurs at the vertex $x_{s_{k}}$ is placed, but the following arguments hold for any other choice of this agent.) 
Recall that $y_{s_{k-1}}$ must be placed next to $z_{s_{k-1}}$, and $c_{s_{k-1}}$ and $f_{s_{k-1}}$ must be placed next to $z_{s_{k-1}}$ or must be isolated.
Similarly, $z_{s_k}$ must be placed next to $y_{s_k}$, and $b_{s_k}$ and $e_{s_k}$ must be placed next to $y_{s_k}$ or must be isolated.
Since it is impossible that these six agents are placed at vertices 1 through 5, at least one of $c_{s_{k-1}}$, $f_{s_{k-1}}$, $b_{s_k}$, and $e_{s_k}$ is isolated.
Now, to make the following argument precise, let us define a {\em bend} as a sequence of three vertices $z_{s_{k-1}}, x_{s_k}, y_{s_k}$, and call $x_{s_k}$ the {\em center} of the bend.
Let us define the {\em bottom-left bend} as the bend whose center lies at the lowest row of the cycle among ones whose centers lie on the leftmost column of the cycle (see Figure \ref{fig:16}(ii)).
Similarly, define the {\em top-right bend} as the bend whose center lies at the highest row among ones whose centers lie on the rightmost column.
By the choice of these special bends, it is easy to check that three vertices in the bottom-left bend and three vertices in the top-right bend do not share a common vertex.
Hence each bend produces at least one isolated agent, but this is a contradiction because at most one agent can be isolated as we have seen before.
Thus we have excluded the existence of a cycle and can conclude that $G'$ consists of one path $u_{s_{1}}, u_{s_{2}}, \ldots, u_{s_{n}}$.
By construction of $G'$, for each $i (1\leq i \leq n-1$), $z_{s_{i}}$ is placed next to $x_{s_{i+1}}$ in $\pi$. Thus, by the aforementioned property, there is an arc $(v_{s_{i}}, v_{s_{i+1}})$ in $G$ and hence there is a Hamiltonian path $v_{s_{1}}, v_{s_{2}}, \ldots, v_{s_{n}}$. 
This completes the proof.

\begin{figure}[htbp]
\begin{center}
\includegraphics[width=150mm]{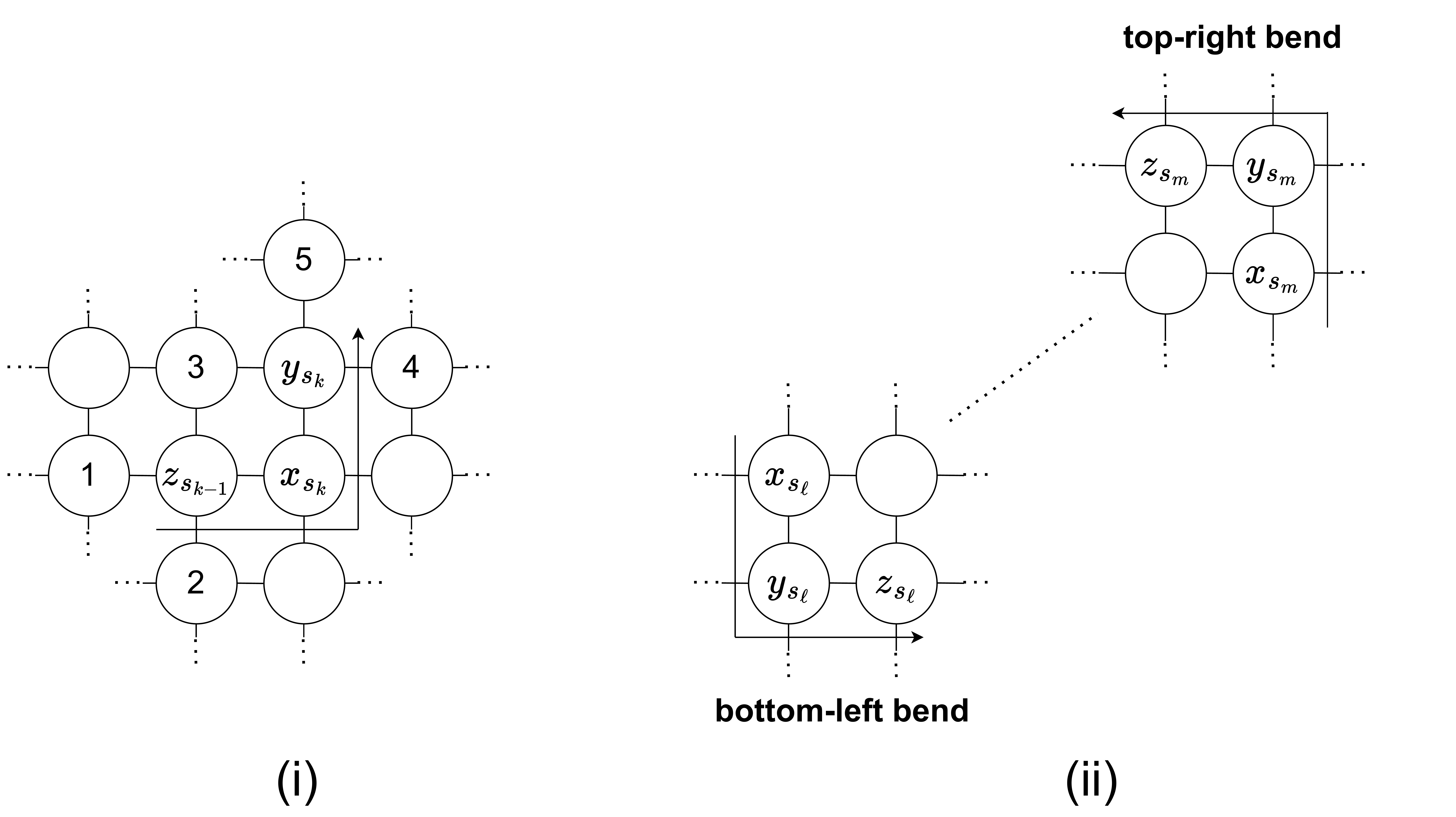}
\caption{Examples of a bend, the bottom-left bend, and the top-right bend}\label{fig:16}
\end{center}
\end{figure}
\end{proof}

\section{Conclusion}
In this paper, we studied the computational complexity of EFA and ESA on $\ell \times m$ grid graphs and showed that both problems are NP-complete.
One of future research directions would be to seek for conditions on preferences when the problems on grid graphs become polynomially solvable.

\section*{Acknowledgments}
 
This work was partially supported by JSPS KAKENHI Grant Numbers JP20K11677 and JP24H00696, and the joint project of Kyoto University and Toyota Motor Corporation, titled ``Advanced Mathematical Science for Mobility Society''.

\end{document}